\documentclass[11pt]{article}

\usepackage[T1]{fontenc}
\usepackage[utf8]{inputenc}
\usepackage{lmodern}
\usepackage{amsmath,amssymb,amsthm,mathtools}
\usepackage{booktabs}
\usepackage{enumitem}
\usepackage{hyperref}
\usepackage{cleveref}
\usepackage[a4paper,margin=1in]{geometry}
\usepackage{authblk}
\usepackage[numbers,sort&compress]{natbib} % Compresses [1,2,3] into [1-3]

\hypersetup{
  colorlinks=true,
  linkcolor=blue,
  citecolor=blue,
  urlcolor=blue
}

\usepackage{graphicx}%
\usepackage{multirow}%
\usepackage{amsfonts}%
\usepackage{mathrsfs}%
\usepackage[title]{appendix}%
\usepackage{xcolor}%
\usepackage{textcomp}%
\usepackage{manyfoot}%
\usepackage{textgreek} % for direct use of Greek letters like α in text
\usepackage{algorithm}%
\usepackage{algorithmic}
\usepackage{listings}%
\usepackage{tikz}
\usetikzlibrary{arrows.meta,positioning}

\newtheorem{theorem}{Theorem}%  meant for continuous numbers
\newtheorem{example}{Example}%
\newtheorem{remark}{Remark}%
\newtheorem{definition}{Definition}%

\title{Shortest Path Problem with Subnormal Gaussian Fuzzy Costs}

\author[1]{Hande G\"{u}nay Akdemir}
\affil[1]{Department of Mathematics, Giresun University, T\"{u}rkiye}
\affil[1]{\texttt{hande.akdemir@giresun.edu.tr}} 
\date{}

\author[2]{Murat Moran}
\affil[2]{Department of Computer Engineering, Giresun University, T\"{u}rkiye}
\affil[2]{\texttt{murat.moran@giresun.edu.tr}}

\begin{document}
\maketitle

%%==================================%%
%% Sample for unstructured abstract %%
%%==================================%%

\abstract{This paper addresses the fuzzy shortest path problem in directed graphs, where edge costs are modeled as generalized fuzzy numbers with Gaussian membership functions. We interpret height as an indicator of information reliability. Based on this view, we introduce a weighted geometric mean to aggregate heights during the addition of generalized Gaussian fuzzy numbers. We employ a reliability-aware ranking that jointly considers the core, height, and standard deviation of fuzzy edge costs to determine the shortest path, thereby capturing their central tendency, reliability, and variability while keeping Dijkstra-level complexity per relaxation. The method yields routes that are not only cost-efficient but also supported by highly reliable information. To assess robustness, we construct a crisp baseline from the ranking and conduct Monte Carlo alpha-cut sampling--drawing membership levels uniformly and then sampling within the induced intervals--to recompute path costs and quantify sensitivity via the mean percentage deviation and its standard deviation. Finally, a large-scale case study on the FAA air traffic network demonstrates that the proposed GGFN--SPP framework scales efficiently to real-world networks, balances cost and reliability through $\alpha$--cut aggregation and risk-aware ranking, and exhibits stable performance under Monte Carlo simulations with subnormal fuzzy costs.}

\vspace{0.3cm} % Adds clean spacing after the abstract

\maketitle

\section{Introduction}
The Shortest Path Problem (SPP) is a fundamental challenge in graph theory, essential for modeling discrete systems like transportation and communication networks. Among the canonical problems in graph theory, the SPP seeks a path of minimum total cost between a source and a terminal (target) node. In the classical setting, edge costs are deterministic. In many applications, however, these costs are subject to epistemic uncertainty (e.g., expert assessments), for which fuzzy numbers offer a natural representation. For instance, travel times between cities may be imprecise due to road or weather conditions, so edge costs can be modeled as fuzzy numbers. This study focuses solely on fuzzy \emph{edge costs}, assuming a known and fixed network topology.

Parametric families such as triangular, trapezoidal, and Gaussian retain their functional form under basic operations (e.g., translation, scaling, and, under mild conditions, addition) due to their parametric closure. Generalized fuzzy numbers are adopted to model epistemic uncertainty, especially when expert assessments may be only partially reliable, that is, when the normality condition is relaxed. In a possibilistic reading, the height of a subnormal fuzzy number acts as a degree of consistency (information reliability) of the underlying evidence \cite{dubois1994possibilistic}. Generalized forms are therefore crucial in real--world settings where the available evidence may be incomplete or partially inconsistent. Gaussian forms, in particular, are attractive in both theory and practice thanks to their smooth analytical structure, compact parametrization, and closed--form $\alpha$--cuts, which facilitate arithmetic and ranking.

Despite their theoretical advantages, subnormal Gaussian membership functions have received limited attention in the literature. Existing treatments typically rely on the extension principle and $\alpha$--cuts; the height parameter is carried along but is rarely \emph{aggregated} within arithmetic operations. When explicit height aggregation is attempted, it typically appears in intuitionistic or picture fuzzy frameworks, whereas in generalized Type--1 settings the combination of membership information is usually delegated to t--norm operators (most often the minimum) rather than to a dedicated height--aggregation rule. In Type--1 settings, height has more often been incorporated at the \emph{ranking} stage. For triangular/trapezoidal fuzzy numbers, Aydin and Akdemir~\cite{aydin-akdemir-2025}--in a proceedings study on generalized trapezoidal fuzzy transportation problems--aggregate heights via a compensatory, weighted harmonic mean with spread--proportional (non--normalized) weights; while the aggregate remains within the input bounds, associativity under repeated addition is not preserved. Within interval Type--2 fuzzy sets, a more recent proposal~\cite{fan2025height} defined unified height update rules for several arithmetic operators and introduced a tuning parameter to encode risk attitudes (from cautious to adventurous). For Gaussian fuzzy numbers specifically, Sen et al.~\cite{sen2021similarity} developed a similarity measure that jointly accounts for mean, standard deviation, and height via an exponential formulation. Tahayori et al.~\cite{tahayori2013shadowed} proposed a construction of shadowed sets from fuzzy sets using both normal and subnormal Gaussian memberships; their method preserves the overall amount of fuzziness via gradual numbers and provides closed--form threshold formulas, ensuring consistency in the representation of fuzzy information.

Despite extensive research on fuzzy SPPs--most of which adapt classical algorithms or employ metaheuristics and rely on triangular or trapezoidal models--studies that relax the normality assumption and explicitly model Gaussian membership functions in Type-V fuzzy graphs remain scarce. For instance, Yao and Lin~\cite{yao2003fuzzy} proposed two fuzzy shortest path models in which edge costs are represented by level$-\lambda$ triangular fuzzy numbers and level$-(1-\beta, 1-\alpha)$ interval-valued fuzzy numbers. The problem is solved using a signed-distance-based ranking method combined with dynamic programming. Elizabeth and Sujatha~\cite{elizabeth2012lambda} introduced novel algorithms for the fuzzy SPP where edge costs are represented as level$-\lambda$ triangular $LR$ fuzzy numbers--an approximation that implicitly captures subnormal behavior by allowing a reduced maximum membership grade. Kavian et al.~\cite{kavian2012twostage} proposed a two-stage fuzzy logic-based scheme to accommodate traffic demand uncertainty in optical core networks. Their approach models future demand fluctuations using Gaussian fuzzy membership functions and integrates a Mamdani-type inference system with Dijkstra-based routing to assign link costs adaptively according to bandwidth availability and uncertainty levels. Hassanzadeh et al.~\cite{hassanzadeh2013genetic} proposed a fuzzy shortest path algorithm that integrates both trapezoidal and Gaussian fuzzy edge lengths, all assumed to be normal (i.e., with maximum membership grade equal to one). Their approach utilizes $\alpha-$cut representations and a least squares-based approximation to construct membership functions for path costs, and applies a genetic algorithm to efficiently solve large-scale instances. Verma and Shukla~\cite{verma2013greedy} proposed a greedy algorithm for the fuzzy SPP, in which they employed quasi-Gaussian fuzzy numbers as edge costs. Their motivation for using Gaussian-shaped membership functions lies in their continuity, differentiability, and compact parametric representation, offering advantages over traditional triangular or trapezoidal fuzzy models. Anusuya and Kavitha~\cite{anusuya2015genetic} proposed a genetic algorithm for the fuzzy SPP, where edge costs are represented by generalized trapezoidal fuzzy numbers. Their model emphasizes the role of each genetic operator individually and incorporates a distance-based fitness function that accounts for fuzziness-related parameters such as mode, divergence, and spreads. Kumar et al.~\cite{kumar2019shortest} proposed a unified fuzzy linear programming approach for solving fuzzy SPPs involving both Type-I and Type-II weighted triangular and trapezoidal fuzzy numbers. Their method enables the prediction of the crisp shortest path length while preserving the original uncertainty structure through non-normal fuzzy membership functions. Dey et al.~\cite{dey2021imprecise} proposed a fuzzy extension of Dijkstra's algorithm to solve the SPP in imprecise environments, where edge lengths are represented by generalized trapezoidal fuzzy numbers. Their approach incorporates a parameterized defuzzification method based on $\lambda$-integral values, allowing decision-makers to adjust the solution based on varying optimism levels. Valdes et al.~\cite{valdes2021fuzzy} modeled communication networks as Type-V fuzzy graphs with link costs represented by normalized triangular fuzzy numbers. They proposed a Fuzzy Dijkstra Algorithm supported by $\alpha-$cut arithmetic and a total integral-based ranking method to compute the shortest path under uncertainty, demonstrating its superiority over classical approaches in simulated network topologies. Ebrahimnejad et al.~\cite{ebrahimnejad2021mabc} addressed the SPP under mixed interval-valued fuzzy numbers, developed an $\alpha$--cut--based aggregation scheme and an extended distance function for path comparison, and devised a modified artificial bee colony algorithm to compute the interval-valued membership of the fuzzy shortest path. Biswal et al.~\cite{biswal2022alpha} formulated the $\alpha$-reliable SPP for uncertain time-dependent road networks within uncertainty theory: link travel times were modeled as independent uncertain normal variables, an additive $\alpha$-reliable path length was derived--thus restoring Bellman's principle--and a Dijkstra-type algorithm was implemented and illustrated with a numerical example.
Abbasi and Allahviranloo~\cite{abbasi2023realistic} proposed a novel fuzzy critical path methodology in project scheduling, where activity durations are represented by generalized quasi-geometric fuzzy numbers--specifically, generalized rational fuzzy numbers with subnormal heights. To enable more realistic modeling and reduce computational burden, the authors introduced TA-based fuzzy arithmetic and a new distance-based ranking approach, which were validated through a case study involving airport cargo ground operations. Akdemir et al.~\cite{akdemir2024mincost} addressed minimum cost flow problems under uncertainty by employing trapezoidal fuzzy numbers with subnormal membership functions. Their credibilistic optimization framework utilizes CVaR minimization to obtain risk-averse solutions, avoiding the information loss caused by normalization and supporting arithmetic operations that preserve shape and height. Pratibha and Dangwal~\cite{pratibha2024leastcost} investigated the fuzzy least cost path problem within supply chain networks by modeling uncertain edge costs as generalized hexagonal fuzzy numbers. Their approach integrates a ranking-based defuzzification method with fuzzy dynamic programming to compute optimal routes under uncertainty without requiring normalization, thereby preserving subnormal membership characteristics in the fuzzy input.

Solving the fuzzy SPP--whether pursued algorithmically (e.g., dynamic programming, Dijkstra's or Bellman--Ford) or via a linear programming formulation--requires well-defined operations for addition, scalar multiplication, and ordering of fuzzy edge costs. The objective of this study is to develop a path-additive defuzzification scheme that is applicable to network optimization problems and compatible with both Gaussian and generalized fuzzy frameworks.

In this study, we adopt an $\alpha$-cut--based formulation for aggregating heights in the addition of generalized Gaussian fuzzy numbers (GGFNs). Specifically, the resulting height is estimated as a weighted geometric mean of the input heights, with weights derived from the standard deviations of the input fuzzy numbers. This ensures that the resulting height remains bounded between the minimum and maximum input heights and that the addition operator is associative. To our knowledge, this is the first approach to use a weighted geometric mean for approximating height in GGFN addition, offering a new perspective in fuzzy arithmetic. Related work on height aggregation and Gaussian memberships includes Akdemir~\cite{akdemir2025ggpfns}, which developed generalized picture fuzzy numbers with Gaussian memberships, proposed triangular approximations, and introduced a compensatory height update for decision problems. Distinct from this line of work, the present paper targets shortest paths with Type--1 GGFNs and introduces an $\alpha$--cut-based, $\sigma$--weighted height aggregation together with a risk--averse ranking index.

Building on this formulation and its implications, the present paper makes the following contributions:
\begin{itemize}
	\item We introduce an $\alpha$--cut--based addition for Type--1 GGFNs that aggregates heights via a $\sigma$--weighted geometric mean. The resulting height stays within input bounds and the addition is associative.
	\item We define risk--averse ranking indices $R^{\mathrm{benefit}}$ and $R^{\mathrm{cost}}$ that couple the mean--like core with a reliability--modulated risk term; for normal fuzzy sets, they reduce to a risk--neutral ranking. We prove the desired monotonicity properties.
	\item Within a single, path--additive scheme, we balance three objectives--minimizing cost, maximizing height (information reliability), and minimizing fuzziness/ambiguity (dispersion $\sigma$ modulated by $h$)--via the proposed $\sigma$--weighted height aggregation and risk--averse ranking; on directed test networks, this yields routes that remain cost--efficient under perturbations while systematically prioritizing high--height edges, thereby making the cost--reliability--ambiguity trade--off explicit in practice.
	\item For robustness analysis, we \emph{adopt} the single-value (fuzzy-numerical) simulation of Chanas and Nowakowski \cite{chanas1988single}--sampling a membership level uniformly and then drawing a cost uniformly from the corresponding $\alpha$--cut--to redraw edge costs and measure sensitivity via the mean percentage deviation and its standard deviation. 
    \item We provide large-scale network experiments that demonstrate the method’s scalability—i.e., practical time/memory growth with problem size—on the FAA network, together with systematic deviation profiles across reliability regimes obtained via single-value $\alpha$-cut simulations.
% If we were to add BAY example 
%	\item We provide large-scale network experiments that demonstrate the method’s scalability—i.e., practical time/memory growth with problem size—on the FAA network (and, in code, on a much larger road network), together with systematic deviation profiles across reliability regimes obtained via single-value $\alpha$-cut simulations on FAA.
\end{itemize}

The remainder of the paper is organized as follows:~\ref{secPre} reviews preliminaries on fuzzy numbers and $\alpha$--cuts, outlines the sampling procedure used for error analysis, and formulates fuzzy SPPs on generalized and Gaussian networks.~\ref{secAddition} presents the adopted $\alpha$--cut--based approach to GGFN addition, introduces the height aggregation rule, and verifies associativity on a small illustrative example (the routine inductive proof is omitted for brevity).~\ref{secRanking} defines the risk--averse ranking indices and establishes their monotonicity properties.~\ref{secExperiments} presents a small-scale numerical illustration of the framework.~\ref{secLarge} details large-scale network tests, confirming computational scalability—i.e., practical time and memory growth with problem size—and illustrating how deviation behavior evolves under varying reliability regimes.~\ref{secConclusion} concludes the paper and highlights the directions for future work.

\section{Preliminaries and Problem Setting}\label{secPre}
The main terminology and basic concepts needed for the discussion are as follows. 

\begin{definition}
	A fuzzy set $\widetilde{A}$ on a universe of discourse $X$ can be expressed as
	\begin{eqnarray*}
		\widetilde{A}=\left\{ \left. \langle x,\mu_{\widetilde{A}}(x)\rangle \,\right|\, x\in X \right\},
	\end{eqnarray*}
	where $\mu_{\widetilde{A}}:X\to[0,1]$ represents the degree to which an element $x$ belongs to $\widetilde{A}$.
\end{definition}

\begin{definition}\cite{chen1985ranking}
	Let $p,q,r\in\mathbb{R}, where $ $p<q<r,$ and $h \in\left(
	0,1\right].$ A generalized fuzzy number $\widetilde{A}$ is a fuzzy subset of the real line
	$\mathbb{R},$ whose membership function $\mu_{\widetilde{A}}$ satisfies the
	following conditions
	\begin{description}
		\item[(i)] $\mu_{\widetilde{A}}:
		\mathbb{R}
		\rightarrow\left[  0,h\right]  $ is continuous,
		\item[(ii)] $\mu_{\widetilde{A}}\left(  x\right)  =0$ for $x<p$ or $x>r$,
		\item[(iii)] $\mu_{\widetilde{A}}\left(  x\right)  $ is strictly increasing on
		$\left[  p,q\right]  $ and strictly decreasing on $\left[  q,r\right]  $,
		\item[(iv)] $\mu_{\widetilde{A}}\left(  q\right)  =h$.
	\end{description}
\end{definition}

\begin{remark}
	If normality holds (i.e., $h=1$) \emph{and} the membership function is convex, then $\widetilde{A}$ is a fuzzy number. Otherwise, when $h\in(0,1)$, $\widetilde{A}$ is \emph{subnormal} with maximal membership (height) $h$. In a possibilistic reading, $h$ can be interpreted as a \emph{degree of consistency} of the underlying evidence; hence $h<1$ indicates \emph{partially inconsistent} information and a non--normalized possibility distribution \cite{dubois1994possibilistic}.
\end{remark}

\begin{definition}
	The $\alpha-$level set of $\widetilde{A}$ is defined as 
	\begin{align*}
		\widetilde{A}_\alpha &=\left\{  x\in \mathbb{R} \vert \mu_{\widetilde{A}}\left(
		x\right) \geq \alpha, \alpha \in [0,h] \right\}\\
		&= [\widetilde{A}_L\left(  \alpha \right),\widetilde{A}_R\left(  \alpha \right)],
	\end{align*} 
	where $\widetilde{A}_L(\alpha)$ and $\widetilde{A}_R(\alpha)$ denote the inverse functions of the left and right membership functions, respectively.
\end{definition}

\begin{definition}\cite{sen2021similarity} 
	The membership function of a GGFN, denoted by $\widetilde{A}=\langle(c_{\widetilde{A}},\sigma
	_{\widetilde{A}});h_{\widetilde{A}} \rangle$, is defined as
	\begin{equation*}
		\mu_{\widetilde{A}}(x) = h_{\widetilde{A}} \exp \left( -\tfrac{1}{2} \left(\tfrac{x-c_{\widetilde{A}}}{\sigma_{\widetilde{A}}}\right)^2 \right),
	\end{equation*}
	where $0 < h_{\widetilde{A}} \leq 1$, $c_{\widetilde{A}}$ is the center (core, or $h_{\widetilde{A}}-$level set), and the parameter $\sigma_{\widetilde{A}}>0$ describes the spread around the central value $c_{\widetilde{A}}$. The center value $c_{\widetilde{A}}$ can be interpreted as the most possible value, with a maximum membership degree of $h_{\widetilde{A}}$.
\end{definition}

\begin{remark}
	According to the definitions in \cite{stoklasa2022relationship}, the possibilistic mean value and variance are calculated as $c_{\widetilde{A}}$ and $\sigma_{\widetilde{A}}^2$, respectively. These quantities are independent of the height and are therefore computed in the same way for fuzzy numbers. The parameter $\sigma_{\widetilde{A}}$ is referred to as the standard deviation. 
\end{remark}

\begin{remark}
\label{rem:fuzzynum}
    For a GGFN $\widetilde{A}=\langle (c_{\widetilde{A}},\sigma_{\widetilde{A}});h_{\widetilde{A}} \rangle$, the left and right inverse functions are
	\[
	\widetilde{A}_{L}(\alpha)=c_{\widetilde{A}}-\sigma_{\widetilde{A}}\sqrt{-2\ln\!\left(\tfrac{\alpha}{h_{\widetilde{A}}}\right)},\qquad
	\widetilde{A}_{R}(\alpha)=c_{\widetilde{A}}+\sigma_{\widetilde{A}}\sqrt{-2\ln\!\left(\tfrac{\alpha}{h_{\widetilde{A}}}\right)}.
	\]
	For any $\alpha\in(0,h_{\widetilde{A}}]$, the $\alpha$-level set of $\widetilde{A}$ is the closed interval $[\widetilde{A}_{L}(\alpha),\,\widetilde{A}_{R}(\alpha)]$.
\end{remark}

\subsection{Sampling scheme for error analysis}
These representations capture the lower and upper extremal values of each GGFN interval $\alpha$-cut. In the simulation procedure, we \emph{adopt} the single-value (fuzzy-numerical) simulation of Chanas and Nowakowski~\cite{chanas1988single}, adapting their Procedure~1 to our setting: we first draw a membership level uniformly and then draw a value uniformly from the induced cut, which yields a single crisp realization of a GGFN. Concretely, a uniformly distributed random number $u$ is generated over an appropriate interval, followed by the generation of another random number $v$ from the uniform distribution on the $u$-cut interval. The weighted arithmetic mean of the endpoints of the $u$-cut interval yields a sample of $\widetilde{A}$ as follows
\begin{equation}
	\hat{A}= v\,\widetilde{A}_{L}(u) + (1-v)\,\widetilde{A}_{R}(u),\label{simu}
\end{equation}
where $u\sim U(0,h_{\widetilde{A}})$ and $v\sim U(0,1)$ are independent. This single-value scheme is then used in a Monte Carlo fashion to redraw edge costs for robustness/error analysis.

\subsection{Fuzzy SPP on generalized Gaussian networks}

A \textit{Type--V fuzzy graph} is defined on known vertices and edges, but the edge costs involve uncertainty. Consider an acyclic directed network $\widetilde{G}=(V,E,\widetilde{\mathcal{C}})$, where $V$ is the set of vertices, $E$ is the set of edges, and $\widetilde{\mathcal{C}}$ denotes the admissible set of fuzzy edge costs. Each edge $(i,j)\in E$ has a nonnegative generalized Gaussian fuzzy cost $\widetilde A_{ij}=\langle(c_{ij},\sigma_{ij});\,h_{ij}\rangle$ with $\alpha$--cut $[\widetilde{A}_{ijL}(\alpha),\widetilde{A}_{ijR}(\alpha)]$, where $\widetilde{A}_{ijL}(\alpha)=c_{ij}-\sigma_{ij}\sqrt{-2\ln(\alpha/h_{ij})}$. We impose nonnegativity of the left endpoints at the evaluation levels used in this paper, i.e., there exists $\alpha_\ast\in(0,1]$ such that $\widetilde{A}_{ijL}(\alpha)\ge 0$ for all $\alpha\in[\alpha_\ast,h_{ij}]$ and all $(i,j)\in E$. A sufficient condition is
$\ \sigma_{ij}\le c_{ij}\big/\sqrt{-2\ln(\alpha_\ast/h_{ij})}\ $ for every edge.\\

We can formulate the fuzzy SPP in the following linear programming form:

\begin{align}
	\min \ & \widetilde{f}(x)\;=\;\sum_{(i,j)\in E} \widetilde{A}_{ij}\,x_{ij} \tag{2}\\[2mm]
	\text{s.t.}\quad 
	& \sum_{j} x_{ij}\;-\;\sum_{j} x_{ji}
	\;=\;
	\begin{cases}
		1, & \text{if } i=s,\\
		0, & \text{if } i\neq s,t\ (i=1,\ldots,n),\\
		-1, & \text{if } i=t,
	\end{cases} \tag{3}\\[2mm]
	& x_{ij}\in\{0,1\}\qquad \forall\,(i,j)\in E \tag{4}
\end{align}
with $s$ and $t$ denoting the source and terminal nodes, respectively.

\section{An $\alpha$--Cut--Based Approach to GGFN Addition: Height Aggregation and Associativity}\label{secAddition}

%\label{arop}
We now present a novel $\alpha$--cut--based approach to generalized fuzzy addition, determining the resulting height via numerical approximation. In this scheme, the height is computed as a weighted geometric mean of the input heights, with weights derived from the standard deviations of the corresponding fuzzy numbers. Under the proposed weighting, the addition operator is associative.

When approximation is used, the corresponding formulation implicitly assumes that the arithmetic and geometric means yield sufficiently close values. Accordingly, the approximation performs better when $h_{\widetilde{A}}$ is close to $h_{\widetilde{B}}$.

The use of weighted averaging in determining the resultant height ensures that the value lies between the minimum and maximum of the input heights. From an information reliability perspective, this reflects a compensatory approach that neither overestimates nor underrepresents the influence of each fuzzy component.

\begin{definition}
	\label{DefArOp}  Let $\widetilde{A}=\langle(c_{\widetilde{A}},\sigma
	_{\widetilde{A}});h_{\widetilde{A}} \rangle$ and  $\widetilde{B}%
	=\langle(c_{\widetilde{B}},\sigma_{\widetilde{B}});h_{\widetilde{B}} \rangle$
	be two independent GGFNs and let $k \in\mathbb{R}$. Then, the following
	operations hold.
	
	\begin{itemize}

		\item Addition of two GGFNs
				
		\begin{align}
			\widetilde{A} + \widetilde{B} 
			= \left\langle 
			\left( c_{\widetilde{A}} + c_{\widetilde{B}}, \, \sigma_{\widetilde{A}} + \sigma_{\widetilde{B}} \right);
			\, h_{\widetilde{A}}^{\frac{\sigma_{\widetilde{A}}}{\sigma_{\widetilde{A}} + \sigma_{\widetilde{B}}}} \cdot
			h_{\widetilde{B}}^{\frac{\sigma_{\widetilde{B}}}{\sigma_{\widetilde{A}} + \sigma_{\widetilde{B}}}}
			\right\rangle \label{eq:addition}
		\end{align}

		\item Scalar Multiplication
		\[
		k\widetilde{A}=\left\langle \left(  kc_{\widetilde{A}},\left\vert k\right\vert
		\sigma_{\widetilde{A}}\right)  ;h_{\widetilde{A}}\right\rangle .
		\]
		
	\end{itemize}
\end{definition}

\begin{remark}
		We analyze the additivity of GGFNs using their a-cut representations.
		Consider two GGFNs 
	$\widetilde{A}=\langle (c_{\widetilde{A}},\sigma_{\widetilde{A}});h_{\widetilde{A}}\rangle$ 
	and 
	$\widetilde{B}=\langle (c_{\widetilde{B}},\sigma_{\widetilde{B}});h_{\widetilde{B}}\rangle$. 
	Under the assumption that the Gaussian form is preserved, the corresponding $\alpha$-cuts are given by
	\begin{align*}
		\widetilde{A}_{\alpha} &  =\left[  c_{\widetilde{A}}-\sigma_{\widetilde{A}}
		\sqrt{-2\ln\!\left(  \tfrac{\alpha}{h_{\widetilde{A}}}\right)  },
		\;c_{\widetilde{A}}+\sigma_{\widetilde{A}}\sqrt{-2\ln\!\left(  \tfrac{\alpha}{h_{\widetilde{A}}}\right)  }\right],
		\quad \alpha\in(0,h_{\widetilde{A}}],\\
		\widetilde{B}_{\alpha} &  =\left[  c_{\widetilde{B}}-\sigma_{\widetilde{B}}
		\sqrt{-2\ln\!\left(  \tfrac{\alpha}{h_{\widetilde{B}}}\right)  },
		\;c_{\widetilde{B}}+\sigma_{\widetilde{B}}\sqrt{-2\ln\!\left(  \tfrac{\alpha}{h_{\widetilde{B}}}\right)  }\right],
		\quad \alpha\in(0,h_{\widetilde{B}}],\\
		(\widetilde{A}+\widetilde{B})_{\alpha}
		& =\Big[ c_{\widetilde{A}}+c_{\widetilde{B}}
		-(\sigma_{\widetilde{A}}+\sigma_{\widetilde{B}})
		\sqrt{-2\ln\!\left(\tfrac{\alpha}{h}\right)}, \;
		c_{\widetilde{A}}+c_{\widetilde{B}}
		+(\sigma_{\widetilde{A}}+\sigma_{\widetilde{B}})
		\sqrt{-2\ln\!\left(\tfrac{\alpha}{h}\right)}\Big]
	\end{align*}
where $h$ denotes the unknown height, with $\alpha \in (0,h]$. According to interval arithmetic, the following identity should hold:
	\begin{align*}
		\sigma_{\widetilde{A}}\sqrt{-2\ln\!\left(  \tfrac{\alpha}{h_{\widetilde{A}}}\right)}
		+\sigma_{\widetilde{B}}\sqrt{-2\ln\!\left(  \tfrac{\alpha}{h_{\widetilde{B}}}\right)}
		=(\sigma_{\widetilde{A}}+\sigma_{\widetilde{B}})
		\sqrt{-2\ln\!\left(  \tfrac{\alpha}{h}\right)}.
	\end{align*}
	
	Since enforcing the above equality yields a complicated expression in $\alpha$ in which $\alpha$ does not cancel, we introduce an approximation. Specifically, invoking the inequality between the weighted arithmetic mean and the weighted geometric mean, we approximate the weighted geometric mean by its weighted arithmetic counterpart as follows:
		
	\begin{align*}
		\sqrt{\ln\!\left(\tfrac{h_{\widetilde{A}}}{\alpha}\right)
			\ln\!\left(\tfrac{h_{\widetilde{B}}}{\alpha}\right)}
		\;\simeq\;
		\tfrac{1}{2}\Big(\ln\!\left(\tfrac{h_{\widetilde{A}}}{\alpha}\right)
		+\ln\!\left(\tfrac{h_{\widetilde{B}}}{\alpha}\right)\Big).
	\end{align*}
	This approximation leads to a tractable expression for the height parameter
	\begin{equation*}
		h=h_{\widetilde{A}}^{\frac{\sigma_{\widetilde{A}}}{\sigma_{\widetilde{A}} + \sigma_{\widetilde{B}}}} \cdot
		h_{\widetilde{B}}^{\frac{\sigma_{\widetilde{B}}}{\sigma_{\widetilde{A}} + \sigma_{\widetilde{B}}}}
	\end{equation*} 
	which forms the basis of our approach.
\end{remark}

\begin{remark}
	A similar $\alpha$-cut-based strategy can be applied to symmetric triangular or trapezoidal generalized fuzzy numbers, using a weighted (not necessarily normalized) harmonic mean to aggregate heights. In this case, the weights are taken proportional to the spreads and are not normalized to sum to one. Under this weighting, the resulting addition operator is not associative. For further details, see~\cite{aydin-akdemir-2025}.
\end{remark}

\section{Risk--Averse Ranking Indices: Definition and Monotonicity}\label{secRanking}
In line with the adopted aggregation rule, we introduce the ranking indices $R^{\mathrm{benefit}}$ and $R^{\mathrm{cost}}$; their precise forms are given below. Throughout, $\log$ denotes the base-10 logarithm.

\begin{definition}
	Let $\widetilde{A}=\langle(c_{\widetilde{A}},\sigma_{\widetilde{A}});\,h_{\widetilde{A}}\rangle$ be a GGFN with $c_{\widetilde{A}}>0$, $\sigma_{\widetilde{A}}>0$, and $h_{\widetilde{A}}\in(0,1]$.
	The \emph{benefit-type} and \emph{cost-type} ranking indices are defined, respectively, by
	\[
	R^{\mathrm{benefit}}(\widetilde{A})
	= c_{\widetilde{A}}+\sigma_{\widetilde{A}}\log\!\big(h_{\widetilde{A}}\big),
	\qquad
	R^{\mathrm{cost}}(\widetilde{A})
	= c_{\widetilde{A}}-\sigma_{\widetilde{A}}\log\!\big(h_{\widetilde{A}}\big).
	\]
\end{definition}
\begin{theorem}
	Let $\widetilde{A}=\langle (c_{\widetilde{A}},\sigma_{\widetilde{A}});\;h_{\widetilde{A}}\rangle$ and $\widetilde{B}=\langle (c_{\widetilde{B}},\sigma_{\widetilde{B}});\;h_{\widetilde{B}}\rangle$ denote two independent cost-type GGFNs. Assume $c_{\widetilde{A}},\,c_{\widetilde{B}},\,k>0$. Under the proposed aggregation scheme, the cost-type ranking index $R^{\mathrm{cost}}$ satisfies:
	\begin{description}
		\item[(i)] $R^{\mathrm{cost}}(\widetilde{A} + \widetilde{B}) = R^{\mathrm{cost}}(\widetilde{A}) + R^{\mathrm{cost}}(\widetilde{B})$,
		\item[(ii)] $R^{\mathrm{cost}}(k \widetilde{A}) = k \, R^{\mathrm{cost}}(\widetilde{A})$.
	\end{description}
\end{theorem}

\begin{proof}
	The proofs below are provided for cost-type variables; similar reasoning applies to the benefit-type index.
	
	\begin{description}
		\item[(i)]
		\begin{align*}
			R^{\mathrm{cost}}(\widetilde{A}+\widetilde{B})
			&= c_{\widetilde{A}+\widetilde{B}}
			- \sigma_{\widetilde{A}+\widetilde{B}}\cdot
			\log\!\left(h_{\widetilde{A}+\widetilde{B}}\right) \\
			&= (c_{\widetilde{A}}+c_{\widetilde{B}})
			- (\sigma_{\widetilde{A}}+\sigma_{\widetilde{B}})\cdot
			\log\!\left(
			h_{\widetilde{A}}^{\frac{\sigma_{\widetilde{A}}}{\sigma_{\widetilde{A}}+\sigma_{\widetilde{B}}}}
			\,h_{\widetilde{B}}^{\frac{\sigma_{\widetilde{B}}}{\sigma_{\widetilde{A}}+\sigma_{\widetilde{B}}}}
			\right) \\
			&= (c_{\widetilde{A}}+c_{\widetilde{B}})
			- (\sigma_{\widetilde{A}}+\sigma_{\widetilde{B}})
			\left[
			\frac{\sigma_{\widetilde{A}}}{\sigma_{\widetilde{A}}+\sigma_{\widetilde{B}}}\log h_{\widetilde{A}}
			+ \frac{\sigma_{\widetilde{B}}}{\sigma_{\widetilde{A}}+\sigma_{\widetilde{B}}}\log h_{\widetilde{B}}
			\right] \\
			&= \big(c_{\widetilde{A}}-\sigma_{\widetilde{A}}\log h_{\widetilde{A}}\big)
			+ \big(c_{\widetilde{B}}-\sigma_{\widetilde{B}}\log h_{\widetilde{B}}\big) \\
			&= R^{\mathrm{cost}}(\widetilde{A}) + R^{\mathrm{cost}}(\widetilde{B}).
		\end{align*}
		
		\item[(ii)]
		\begin{align*}
			R^{\mathrm{cost}}(k\widetilde{A})
			&= c_{k\widetilde{A}} - \sigma_{k\widetilde{A}}\cdot \log\!\left(h_{k\widetilde{A}}\right) \\
			&= k c_{\widetilde{A}} - k \sigma_{\widetilde{A}}\cdot \log\!\left(h_{\widetilde{A}}\right)
			\qquad (\text{since }k>0) \\
			&= k\,R^{\mathrm{cost}}(\widetilde{A}).
		\end{align*}
	\end{description}
\end{proof}

\begin{remark}
	In defining the ranking functions, their monotonicity can be verified by analyzing partial derivatives with respect to the underlying parameters. Assume $c>0$, $\sigma>0$, and $h\in(0,1]$.
	
	\medskip
	\noindent\textbf{Benefit-type variable.} Let $R^{\mathrm{benefit}}(\langle (c,\sigma);\;h\rangle)=c+\sigma\log(h)$. Then
	\[
	\frac{\partial R^{\mathrm{benefit}}}{\partial c}=1>0,\qquad
	\frac{\partial R^{\mathrm{benefit}}}{\partial h}=\frac{\sigma}{h\ln10}>0,\qquad
	\frac{\partial R^{\mathrm{benefit}}}{\partial \sigma}=\log(h)\le 0.
	\]

    In the benefit-type index, higher values of $c$ and $h$ improve the score, while volatility ($\sigma$) reduces it. Conversely, for the cost-type index, $c$ and $\sigma$ act as penalties, while $h$ mitigates the cost, reflecting the preference for high-reliability information.
    	
	\medskip
	\noindent\textbf{Cost-type variable.} Let $R^{\mathrm{cost}}(\langle (c,\sigma);\;h\rangle)=c-\sigma\log(h)$ and interpret lower values as more preferred (``less is better''). Then
	\[
	\frac{\partial R^{\mathrm{cost}}}{\partial c}=1>0,\qquad
	\frac{\partial R^{\mathrm{cost}}}{\partial h}=-\frac{\sigma}{h\ln10}<0,\qquad
	\frac{\partial R^{\mathrm{cost}}}{\partial \sigma}=-\log(h)\ge 0.
	\]
	Hence, larger $c$ and $\sigma$ worsen (increase) the index, whereas larger $h$ improves (decreases) it.
	
	\medskip
	In summary, under $c>0$, $\sigma>0$, and $h\in(0,1]$, both $R^{\mathrm{benefit}}$ and $R^{\mathrm{cost}}$ exhibit the desired monotonicity, reflecting the principles that ``more is better'' for benefits and ``less is better'' for costs. Conceptually, the indices are \emph{mean--risk} measures with a risk-averse orientation: they couple the mean-like core $c$ with a risk adjustment that scales with $\sigma$ and whose magnitude is controlled by $|\log h|$. For normal fuzzy sets ($h=1$), the adjustment vanishes and the index is risk-neutral, yielding $R^{\mathrm{benefit}}=R^{\mathrm{cost}}=c$. As reliability decreases ($h\downarrow$) or variability increases ($\sigma\uparrow$), the penalty magnitude $\sigma|\log h|$ grows, pushing the score in a less preferred direction--\emph{increasing} for cost-type and \emph{decreasing} for benefit-type. In other words, from a risk-averse (pessimistic) perspective--and relative to risk-neutral defuzzification based solely on the core $c$--cost-type fuzzy quantities are defuzzified to larger crisp values, whereas benefit-type quantities are crispified to smaller ones.
    
    The cost-type index can be generalized by introducing a risk--attitude parameter $\kappa$ as
\begin{equation}
  R^{\mathrm{cost}}(\langle (c,\sigma);\;h\rangle,\kappa)
  =
  c - \kappa\,\sigma \log(h),
  \label{eq:ranking_cost}
\end{equation}
where $\kappa \ge 0$. The baseline specification is recovered by setting $\kappa = 1$, while larger values of $\kappa$ penalize alternatives with smaller heights $h$, even if they have slightly lower cores.
\end{remark}

We illustrate the height aggregation rule and the ranking index in~\ref{exampl1} (for both benefit- and cost-type cases) and, in this instance, confirm that the resulting addition operator is associative; a general proof by induction is straightforward and is left to the reader.

\begin{example}\label{exampl1}
~\ref{tab:example1} reports the core, standard deviation, and height values for 
	$\widetilde{A}=\langle(15,3);\,0.6\rangle$, 
	$\widetilde{B}=\langle(5,1);\,0.7\rangle$, 
	$\widetilde{C}=\langle(5,1);\,0.9\rangle$ 
	and their combinations under the adopted approach, together with the corresponding ranking scores for both benefit-type and cost-type settings. The ranking results demonstrate that:
	\[
	\widetilde{B} \prec \widetilde{C}, \quad 
	\widetilde{A}+\widetilde{B} \prec \widetilde{A}+\widetilde{C}, \quad 
	2\widetilde{B} \prec \widetilde{B}+\widetilde{C} \prec 2\widetilde{C}.
	\]
	These results align with the expected preference ordering. 
	Note that $\widetilde{B} \prec \widetilde{C}$ indicates a preference for $\widetilde{C}$ over $\widetilde{B}$.
   
\begin{table}[htbp]
	\caption{Parameters (core $c$, standard deviation $\sigma$, and height $h$) and ranking indices for $\widetilde{A}$, $\widetilde{B}$, $\widetilde{C}$ and their combinations (\emph{All non-integer values are reported to four decimal places.})}
	\label{tab:example1}
    \begin{tabular*}{\textwidth}{@{\extracolsep\fill}cccccc}
		\toprule
		& \multicolumn{3}{c}{Parameters} & \multicolumn{2}{c}{Ranking indices} \\
		\cmidrule(lr){2-4}\cmidrule(lr){5-6}
		& Core $c$ & Std.\ Dev.\ $\sigma$ & Height $h$ & $R^{\mathrm{benefit}}$ & $R^{\mathrm{cost}}$ \\
		\midrule
		$\widetilde{A}$               & 15 & 3 & 0.6000 & 14.3345 & 15.6655 \\
		$\widetilde{B}$               & 5  & 1 & 0.7000 &  4.8451 &  5.1549 \\
		$\widetilde{C}$               & 5  & 1 & 0.9000 &  4.9542 &  5.0458 \\
		$\widetilde{A}+\widetilde{B}$ & 20 & 4 & 0.6236 & 19.1796 & 20.8204 \\
		$\widetilde{A}+\widetilde{C}$ & 20 & 4 & 0.6640 & 19.2887 & 20.7113 \\
		$2\widetilde{B}$              & 10 & 2 & 0.7000 &  9.6902 & 10.3098 \\
		$\widetilde{B}+\widetilde{C}$ & 10 & 2 & 0.7937 &  9.7993 & 10.2007 \\
		$2\widetilde{C}$              & 10 & 2 & 0.9000 &  9.9085 & 10.0915 \\
        \bottomrule
	\end{tabular*}
\end{table}

Under the adopted approach, all parenthesizations coincide:
\[
(\widetilde A+\widetilde B)+\widetilde C
= \left\langle (25,5);\;
(0.6235739)^{\frac{4}{5}}\,(0.9)^{\frac{1}{5}}\right\rangle
= \left\langle (25,5);\,0.6710561\right\rangle,
\]
\[
\widetilde A+(\widetilde B+\widetilde C)
= \left\langle (25,5);\;
(0.7937254)^{\frac{2}{5}}\,(0.6)^{\frac{3}{5}}\right\rangle
= \left\langle (25,5);\,0.6710561\right\rangle,
\]
\[
(\widetilde A+\widetilde C)+\widetilde B
= \left\langle (25,5);\;
(0.6640092)^{\frac{4}{5}}\,(0.7)^{\frac{1}{5}}\right\rangle
= \left\langle (25,5);\,0.6710561\right\rangle.
\]

For completeness, the three-term sum (without parentheses) is
\[
\widetilde A+\widetilde B+\widetilde C
= \left\langle(25,5);\;0.6^{\frac{3}{5}}\,0.7^{\frac{1}{5}}\,0.9^{\frac{1}{5}}\right\rangle
= \left\langle(25,5);\,0.6710561\right\rangle.
\]
We omit further details here, as the relevant intermediate calculations are already provided in~\ref{tab:example1}.
\end{example}

\section{Numerical Experiments: Fuzzy SPP Cases}\label{secExperiments}

This section reports numerical experiments conducted under generalized Gaussian fuzzy edge costs. We focus on the impact of cost uncertainty and quantify the error induced by replacing the original fuzzy model with crisp counterparts.

\paragraph{Stage 1 (model solution)}
We first solve the proposed ranking index-based minimization model under GGFN costs and record the resulting optimal decision (path) and objective value.

\paragraph{Terminology and rationale}
\emph{Ex-ante} refers to the baseline path chosen \emph{before} sampling, using the reliability-aware (risk-averse) ranking while uncertainty is still present. 
\emph{Ex-post} refers to, for each Monte Carlo draw, the path obtained \emph{after} costs are realized as crisp values by re-optimizing the deterministic problem purely by cost. 
In other words, ex-ante uses the membership information (including heights) to pick a conservative baseline, whereas ex-post treats a sampled instance as fully known and selects the cost-minimizing path for that instance.

Our goal is not to reproduce the ex-post minimum, but to assess how closely a reliability-aware solution tracks ex-post optima and to quantify the associated \emph{reliability premium} via Monte Carlo $\alpha$-cut sampling.

\paragraph{Stage 2 (simulation for error analysis)}
To assess the performance loss due to defuzzification, we generate crisp cost instances from the GGFNs and re-solve the induced deterministic shortest-path problems. In particular, for each edge $e$ we draw a level $u \sim U(0,h_e)$ and a mixing weight $v \sim U(0,1)$, compute the $\alpha$-cut endpoints $c_{e,L}(u)$ and $c_{e,R}(u)$ of the GGFN, and set the sample cost $\hat c_e$ as in~\ref{simu}.

This single-value scheme follows Chanas and Nowakowski's fuzzy-numerical simulation (adapted to GGFNs), in which membership levels are drawn uniformly and values are sampled uniformly from the induced $\alpha$-cuts~\cite{chanas1988single}. The resulting crisp instances are solved repeatedly to produce meaningful Monte Carlo summaries for evaluating the risk-averse solution. To ensure valid edge costs, we use rejection sampling: if the generator produces a negative cost, we discard it and draw again.

\paragraph{Error metric}
Let $\bar{\boldsymbol c}$ denote the \emph{baseline crisp} edge coefficients obtained from the ranking rule, and let $\hat{\boldsymbol c}$ denote the \emph{sampled crisp} edge coefficients obtained by $\alpha$-cut sampling for a given instance. For each instance, let $z^{0}$ be the optimal objective under $\bar{\boldsymbol c}$, and let $z$ be the optimal objective under $\hat{\boldsymbol c}$. We report the absolute percentage deviation from the baseline
\begin{equation}\label{eq:mc-error}
	\mathrm{Dev}(\%) \;=\; 100\,\frac{\lvert z - z^{0}\rvert}{\,z^{0}\,},
\end{equation}
defined whenever $z^{0}>0$. We then present the empirical distribution of Dev (\ref{eq:mc-error}) across replications and summarize it by the mean percentage deviation and the standard deviation of the percentage deviation.

\paragraph{Implementation}
All experiments were executed on a standard desktop environment. The deterministic subproblems can be solved with a shortest-path routine (e.g., Dijkstra) or a linear programming solver. In practice, we used MATLAB built-in graph routines (\texttt{graph}/\texttt{digraph}) and the \texttt{shortestpath} function to solve the induced shortest-path instances. For reproducibility of results, we use the random seed (\texttt{rng(42)} in MATLAB). For each case, we perform $10$ outer replications, each comprising $N=1000$ Monte Carlo samples. For every sample, levels are drawn independently across edges ($u_e\sim U(0,h_e)$, guarded against $u_e=0$ by machine-epsilon clipping), and mixing weights are drawn as $v_e\sim U(0,1)$. All draws are independent across edges and repetitions. 

\begin{example}
	The following example is adapted from Hasuike~\cite{hasuike2013robust}, who studied a robust shortest--path model with random edge costs. We take the reported mean--variance pairs as $(c,\sigma^{2})$ and represent each edge cost as a GGFN $\langle(c,\sigma);\,h\rangle$ with $\sigma=\sqrt{\sigma^{2}}$. 
	
	To emulate different reliability regimes, edge heights are drawn independently as
	$\{h_{ij}\}_{(i,j)\in E}{\sim}\mathrm{Beta}(a,b)$, where $a,b>0$ are the Beta shape parameters; We primarily evaluate four settings: \textsc{high} $(a,b)=(8,2)$, \textsc{moderate} $(4,3)$, \textsc{low} $(2,5)$, and a \textsc{mixed} regime defined by
	$h_{ij}\sim (1-\varepsilon)\,\mathrm{Beta}(8,2)+\varepsilon\,\mathrm{Beta}(2,5)$ (with $\varepsilon=0.2$). For completeness, two additional mixed settings used in our sensitivity checks are reported in~\ref{tab:heightsets}.
	
	\ref{fig:sampleNetworkGGFN} depicts the resulting network; the topology and $(c,\sigma^{2})$ pairs follow Hasuike~\cite{hasuike2013robust}, while the heights shown are a single i.i.d. draw from the \textsc{high} setting $\mathrm{Beta}(8,2)$.
	
	\begin{figure}[htbp]
		\centering
		\resizebox{0.9\linewidth}{!}{
		\begin{tikzpicture}[
			x=3.0cm,y=2.2cm,
			>=Latex,
			vtx/.style={circle,draw,minimum size=8mm,inner sep=0pt,font=\small},
			lbl/.style={fill=white,inner sep=1pt,font=\scriptsize,draw=none}
			]
			
			%--- nodes
			\node[vtx] (A) at (0, 0) {A};
			\node[vtx] (B) at (1, 1) {B};
			\node[vtx] (C) at (1,-1) {C};
			\node[vtx] (D) at (2, 0) {D};
			\node[vtx] (E) at (3, 1) {E};
			\node[vtx] (F) at (3,-1) {F};
			\node[vtx] (G) at (4, 0) {G};
			
			%--- edges with labels <(c, sigma); h> (2 decimals)
			\draw[->] (A) -- (B) node[pos=0.35,above left, lbl] {$\langle(15,\,2.24);\;0.89\rangle$};
			\draw[->] (A) -- (C) node[pos=0.35,below left, lbl] {$\langle(18,\,1.41);\;0.97\rangle$};
			
			\draw[->] (B) -- (D) node[pos=0.52,left,       lbl] {$\langle(19,\,2.83);\;0.88\rangle$};
			\draw[->] (C) -- (D) node[pos=0.48,below,      lbl] {$\langle(18,\,1.41);\;0.92\rangle$};
			
			\draw[->] (B) -- (E) node[pos=0.55,above,      lbl] {$\langle(33,\,4.58);\;0.66\rangle$};
			
			% -- adjusted around E to avoid overlap --
			\draw[->] (D) -- (E) node[pos=0.45,above right=4pt, lbl] {$\langle(17,\,1.73);\;0.90\rangle$};
			\draw[->] (E) -- (G) node[pos=0.60,below=3pt,      lbl] {$\langle(12,\,2.00);\;0.80\rangle$};
			
			% -- adjusted around F to avoid overlap --
			\draw[->] (D) -- (F) node[pos=0.35,below right=4pt,lbl] {$\langle(10,\,1.41);\;0.79\rangle$};
			\draw[->] (C) -- (F) node[pos=0.55,below,         lbl] {$\langle(35,\,3.16);\;0.94\rangle$};
			\draw[->] (F) -- (G) node[pos=0.80,below=3pt,     lbl] {$\langle(18,\,2.24);\;0.85\rangle$};
			
		\end{tikzpicture}
	}
		\caption{Sample network with GGFN edge costs $\langle(c,\sigma);h\rangle$.
			Here $\sigma=\sqrt{\text{Variance}}$ from the original labels, and heights are sampled i.i.d. from $\mathrm{Beta}(8,2)$ (high reliability). All decimal quantities are reported to two decimal places.}
		\label{fig:sampleNetworkGGFN}
	\end{figure}
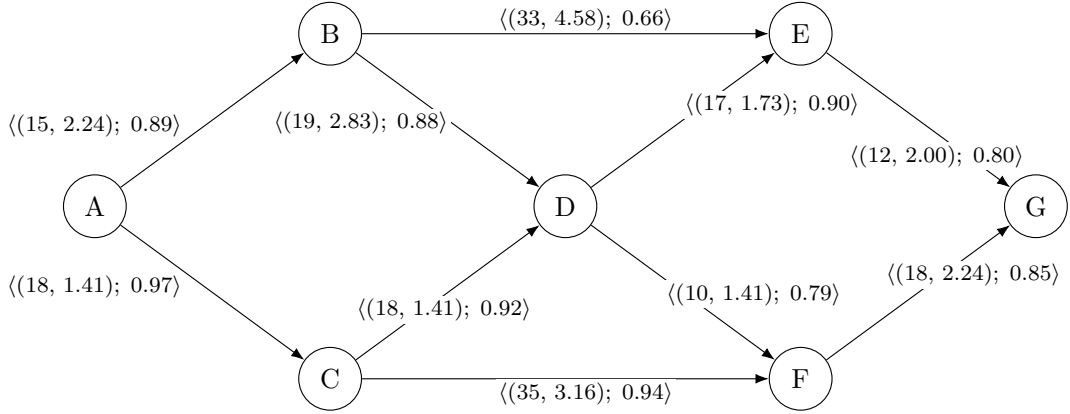
		
	\begin{table}[htbp]
		\caption{Example height sets $h_{ij}$ under four reliability regimes plus two mixed settings. 
			\textsc{High}: $\mathrm{Beta}(8,2)$; \textsc{Moderate}: $\mathrm{Beta}(4,3)$; \textsc{Low}: $\mathrm{Beta}(2,5)$; 
			\textsc{Mixed}: $(1-\varepsilon)\,\mathrm{Beta}(8,2)+\varepsilon\,\mathrm{Beta}(2,5)$ (here $\varepsilon=0.2$).
			\textsc{Mixed--A} and \textsc{Mixed--B} report two alternative mixed settings.}\label{tab:heightsets}
		\begin{tabular}{lcccccc}
			\toprule
			Edge & \textsc{High} & \textsc{Moderate} & \textsc{Low} & \textsc{Mixed} & \textsc{Mixed--A} & \textsc{Mixed--B} \\
			\midrule
			AB & 0.89 & 0.51 & 0.05 & 0.33 & 0.20 & 0.98 \\
			AC & 0.97 & 0.71 & 0.09 & 0.77 & 0.98 & 0.20 \\
			BD & 0.88 & 0.89 & 0.20 & 0.97 & 0.50 & 0.98 \\
			BE & 0.66 & 0.59 & 0.26 & 0.63 & 0.20 & 0.20 \\
			CD & 0.92 & 0.34 & 0.16 & 0.85 & 0.98 & 0.20 \\
			CF & 0.94 & 0.75 & 0.44 & 0.80 & 0.50 & 0.20 \\
			DE & 0.90 & 0.77 & 0.24 & 0.25 & 0.50 & 0.20 \\
			DF & 0.79 & 0.65 & 0.20 & 0.78 & 0.98 & 0.98 \\
			EG & 0.80 & 0.61 & 0.49 & 0.76 & 0.20 & 0.20 \\
			FG & 0.85 & 0.58 & 0.09 & 0.77 & 0.98 & 0.98 \\
			\bottomrule
		\end{tabular}
    \end{table}

\begin{table}[htbp]
	\caption{Path summaries (part I). Core $c$, dispersion $\sigma$, and rankings under \textsc{High/Moderate/Low} reliability settings.}
	\label{tab:paths-part1}
	\begin{tabular*}{\textwidth}{@{\extracolsep{\fill}}l r r rr rr rr}
		\toprule
		& & & \multicolumn{2}{c}{\textsc{High} $(a,b)=(8,2)$} & \multicolumn{2}{c}{\textsc{Moderate} $(4,3)$} & \multicolumn{2}{c}{\textsc{Low} $(2,5)$}\\
		\cmidrule(lr){4-5}\cmidrule(lr){6-7}\cmidrule(lr){8-9}
		Path & $c$ & $\sigma$ & $h$ & $R^{\mathrm{cost}}$ & $h$ & $R^{\mathrm{cost}}$ & $h$ & $R^{\mathrm{cost}}$ \\
		\midrule
		$\mathrm{A}\!-\!\mathrm{B}\!-\!\mathrm{E}\!-\!\mathrm{G}$         & 60 & 8.82 & 0.74 & \textbf{61.14} & 0.57 & \textbf{62.13} & 0.20 & \textbf{66.18} \\
		$\mathrm{A}\!-\!\mathrm{B}\!-\!\mathrm{D}\!-\!\mathrm{E}\!-\!\mathrm{G}$ & 63 & 8.80 & 0.87 & 63.54 & 0.69 & 64.43 & 0.18 & 69.54 \\
		$\mathrm{A}\!-\!\mathrm{B}\!-\!\mathrm{D}\!-\!\mathrm{F}\!-\!\mathrm{G}$ & 62 & 8.71 & 0.86 & 62.56 & 0.66 & 63.59 & 0.12 & 70.17 \\
		$\mathrm{A}\!-\!\mathrm{C}\!-\!\mathrm{D}\!-\!\mathrm{E}\!-\!\mathrm{G}$ & 65 & 6.56 & 0.89 & 65.34 & 0.59 & 66.50 & 0.22 & 69.28 \\
		$\mathrm{A}\!-\!\mathrm{C}\!-\!\mathrm{D}\!-\!\mathrm{F}\!-\!\mathrm{G}$ & 64 & 6.48 & 0.84 & 64.37 & 0.71 & 65.66 & 0.29 & 69.91\\
		$\mathrm{A}\!-\!\mathrm{C}\!-\!\mathrm{F}\!-\!\mathrm{G}$         & 71 & 6.81 & 0.88 & 71.26 & 0.79 & 72.14 & 0.62 & 75.93 \\
		\bottomrule
	\end{tabular*}
\end{table}

\begin{table}[htbp]
	\caption{Path summaries (part II). Rankings under \textsc{Mixed/Mixed--A/Mixed--B} reliability settings.}
	\label{tab:paths-part2}
	\begin{tabular*}{\textwidth}{@{\extracolsep{\fill}}l rr rr rr}
		\toprule
		& \multicolumn{2}{c}{\textsc{Mixed} } & \multicolumn{2}{c}{\textsc{Mixed--A}} & \multicolumn{2}{c}{\textsc{Mixed--B}} \\
		\cmidrule(lr){2-3}\cmidrule(lr){4-5}\cmidrule(lr){6-7}
		Path & $h$ & $R^{\mathrm{cost}}$ & $h$ & $R^{\mathrm{cost}}$ & $h$ & $R^{\mathrm{cost}}$ \\
		\midrule
		$\mathrm{A}\!-\!\mathrm{B}\!-\!\mathrm{E}\!-\!\mathrm{G}$
		& 0.56 & \textbf{62.24} & 0.20 & 66.16 & 0.30 & 64.62 \\
		$\mathrm{A}\!-\!\mathrm{B}\!-\!\mathrm{D}\!-\!\mathrm{E}\!-\!\mathrm{G}$
		& 0.53 & 65.42 & 0.32 & 67.33 & 0.50 & 65.65 \\
		$\mathrm{A}\!-\!\mathrm{B}\!-\!\mathrm{D}\!-\!\mathrm{F}\!-\!\mathrm{G}$
		& 0.67 & 63.52 & 0.52 & 64.45 & 0.98 & \textbf{62.08} \\
		$\mathrm{A}\!-\!\mathrm{C}\!-\!\mathrm{D}\!-\!\mathrm{E}\!-\!\mathrm{G}$
		& 0.58 & 66.56 & 0.51 & 66.94 & 0.20 & 69.59 \\
		$\mathrm{A}\!-\!\mathrm{C}\!-\!\mathrm{D}\!-\!\mathrm{F}\!-\!\mathrm{G}$
		& 0.79 & 64.66 & 0.98 & \textbf{64.06} & 0.49 & 66.01 \\
		$\mathrm{A}\!-\!\mathrm{C}\!-\!\mathrm{F}\!-\!\mathrm{G}$
		& 0.79 & 71.71 & 0.72 & 71.98 & 0.34 & 74.22 \\
		\bottomrule
	\end{tabular*}
\end{table}

As summarized in~\ref{tab:paths-part1} and~\ref{tab:paths-part2}, the cost--type ranking selects $\mathrm{A}\!-\!\mathrm{B}\!-\!\mathrm{E}\!-\!\mathrm{G}$ as the optimal path under the \textsc{High, Moderate, Low}, and \textsc{Mixed} regimes, whereas other mixed settings yield $\mathrm{A}\!-\!\mathrm{C}\!-\!\mathrm{D}\!-\!\mathrm{F}\!-\!\mathrm{G}$  (\textsc{Mixed--A}) and $\mathrm{A}\!-\!\mathrm{B}\!-\!\mathrm{D}\!-\!\mathrm{F}\!-\!\mathrm{G}$ (\textsc{Mixed--B}), respectively. Tables~\ref{tab:paths-part1} and~\ref{tab:paths-part2} indicate that the optimal path is stable ($\mathrm{A}\!-\!\mathrm{B}\!-\!\mathrm{E}\!-\!\mathrm{G}$) across broad i.i.d. reliability regimes (\textsc{High/Moderate/Low} and \textsc{Mixed}), and changes only under targeted mixed settings (\textsc{Mixed--A} and \textsc{Mixed--B}). Consistent with our results, Hasuike \cite{hasuike2013robust} reports $\mathrm{A}\!-\!\mathrm{C}\!-\!\mathrm{D}\!-\!\mathrm{F}\!-\!\mathrm{G}$ as optimal under both the probability-maximizing and robust models, whereas the basic SPP selects $\mathrm{A}\!-\!\mathrm{B}\!-\!\mathrm{E}\!-\!\mathrm{G}$. Moreover, as the goal parameter $f_L$ increases from $60$ to $66$, the optimal path shifts from $\mathrm{A}\!-\!\mathrm{B}\!-\!\mathrm{E}\!-\!\mathrm{G}$ ($f_L=60$) to $\mathrm{A}\!-\!\mathrm{B}\!-\!\mathrm{D}\!-\!\mathrm{F}\!-\!\mathrm{G}$ ($f_L=63$) and $\mathrm{A}\!-\!\mathrm{C}\!-\!\mathrm{D}\!-\!\mathrm{F}\!-\!\mathrm{G}$ ($f_L=66$). 

Holding $f_L=66$, $f_U=75$, and $d_U=3$, increasing $d_L$ from $1.0$ to $2.5$ shifts the optimal path from $\mathrm{A}\!-\!\mathrm{B}\!-\!\mathrm{E}\!-\!\mathrm{G}$ ($d_L=1.0$) to $\mathrm{A}\!-\!\mathrm{B}\!-\!\mathrm{D}\!-\!\mathrm{F}\!-\!\mathrm{G}$  ($1.5$) and then to $\mathrm{A}\!-\!\mathrm{C}\!-\!\mathrm{D}\!-\!\mathrm{F}\!-\!\mathrm{G}$  ($2.0,\,2.5$),\,
consistent with the model's design: small $d_L$ places little weight on dispersion (recovering the standard SPP solution), whereas larger $d_L$ emphasizes variance reduction and yields more risk-averse routes.

Overall, the optimal solution is sensitive to parameter variation and the specification of risk preferences; even modest changes in these inputs can reorder candidate paths and shift the optimum toward more risk-averse routes.

Having specified the reliability-aware ranking and obtained the baseline path and objective (Stage~1), we now proceed to Stage~2 to assess robustness: we generate crisp instances via Monte Carlo $\alpha$-cut sampling and summarize performance using the percentage deviation statistics.

\begin{table}[htbp]
	\caption{Absolute percentage deviation (Dev~\%) and \emph{standard deviation of Dev} for \textsc{High} and \textsc{Moderate} reliability settings across $10$ replications (rounded to 4 decimals).}
	\label{tab:dev-high-moderate}
	\begin{tabular}{@{}rccccc@{}}
		\toprule
		& \multicolumn{2}{c}{\textsc{High}} & \multicolumn{2}{c}{\textsc{Moderate}} \\
		\cmidrule(lr){2-3} \cmidrule(lr){4-5}
		Rep. & Dev (\%) & Std. Dev. (Dev) & Dev (\%) & Std. Dev. (Dev) \\
		\midrule
		1  & 5.3192 & 4.5710 & 6.1338 & 4.9280 \\
		2  & 5.2651 & 4.3867 & 6.2886 & 4.8916 \\
		3  & 5.3035 & 4.5718 & 6.3028 & 5.0162 \\
		4  & 5.2564 & 4.4244 & 6.3757 & 4.8692 \\
		5  & 5.3150 & 4.1998 & 6.3480 & 4.6182 \\
		6  & 5.6864 & 4.4516 & 7.5891 & 5.0268 \\
		7  & 5.4206 & 4.4216 & 7.0094 & 5.0284 \\
		8  & 5.2920 & 4.1643 & 6.6642 & 4.6497 \\
		9  & 5.1006 & 4.2974 & 5.6678 & 4.6437 \\
		10 & 5.7035 & 4.3878 & 6.6414 & 4.9216 \\
		\midrule
		\textbf{Mean} & \textbf{5.3662} & \textbf{4.3876} & \textbf{6.5021} & \textbf{4.8593} \\
		\bottomrule
	\end{tabular}
\end{table}

\begin{table}[htbp]
	\caption{Absolute percentage deviation (Dev~\%) and \emph{standard deviation of Dev} for \textsc{Low} and \textsc{Mixed} settings across $10$ replications (rounded to 4 decimals).}
	\label{tab:dev-low-mixed}
	\begin{tabular}{@{}rcccc@{}}
		\toprule
		& \multicolumn{2}{c}{\textsc{Low}} & \multicolumn{2}{c}{\textsc{Mixed}} \\
		\cmidrule(lr){2-3}\cmidrule(lr){4-5}
		Rep. & Dev (\%) & Std. Dev. (Dev) & Dev (\%) & Std. Dev. (Dev) \\
		\midrule
		1  & 8.6731 & 5.4286 & 5.2362 & 4.3685 \\
		2  & 10.9313 & 4.9576 & 6.1633 & 4.7213 \\
		3  & 10.1704 & 5.2453 & 6.6705 & 5.0122\\
		4  & 10.2511 & 5.1000 & 6.8820 & 4.9866 \\
		5  & 9.9290 & 5.1481 & 5.0512 & 4.3193 \\
		6  & 9.9177 & 5.0496 & 5.6908 & 4.4811 \\
		7  & 9.9895 & 5.1199 & 5.2315 & 4.1998 \\
		8  & 9.1868 & 5.2812 & 6.2319 & 4.7699 \\
		9  & 10.0580 & 5.3762 & 6.0984 & 4.7775 \\
		10 & 11.6725 & 5.3328 & 6.6728 & 4.7795 \\
		\midrule
		\textbf{Mean} & \textbf{10.0780} & \textbf{5.2039} & \textbf{5.9929} & \textbf{4.6416} \\
		\bottomrule
	\end{tabular}
\end{table}

\begin{table}[htbp]
	\caption{Absolute percentage deviation (Dev~\%) and \emph{standard deviation of Dev} for \textsc{Mixed--A} and \textsc{Mixed--B} settings across $10$ replications (rounded to 4 decimals).}
	\label{tab:dev-mixedA-mixedB}
	\begin{tabular}{@{}rcccc@{}}
		\toprule
		& \multicolumn{2}{c}{\textsc{Mixed--A}} & \multicolumn{2}{c}{\textsc{Mixed--B}} \\
		\cmidrule(lr){2-3}\cmidrule(lr){4-5}
		Rep. & Dev (\%) & Std. Dev. (Dev) & Dev (\%) & Std. Dev. (Dev) \\
		\midrule
		1  & 8.6738 & 5.1307 & 6.2512 & 4.7823 \\
		2  & 8.8612 & 5.2180 & 6.4449 & 4.8518 \\
		3  & 8.6490 & 5.0848 & 6.2677 & 4.6650 \\
		4  & 8.8014 & 5.2122 & 6.3393 & 4.9146 \\
		5  & 8.8713 & 5.2225 & 6.4709 & 4.8354 \\
		6  & 8.7601 & 5.1232 & 6.4131 & 4.6707 \\
		7  & 9.0235 & 5.4053 & 6.5896 & 5.0781 \\
		8  & 9.2159 & 5.3598 & 6.7255 & 5.0897 \\
		9  & 8.9630 & 5.3791 & 6.5394 & 5.0442 \\
		10 & 8.7775 & 5.1267 & 6.3294 & 4.7843 \\
		\midrule
		\textbf{Mean} & \textbf{8.8597} & \textbf{5.2262} & \textbf{6.4371} & \textbf{4.8716} \\
		\bottomrule
	\end{tabular}
\end{table}

\paragraph{Discussion}
Across settings, deviation increases as reliability decreases. Using Tables~\ref{tab:dev-high-moderate}, \ref{tab:dev-low-mixed}, and \ref{tab:dev-mixedA-mixedB}, the mean Dev (SD of Dev) is \textbf{5.37\%} (\textbf{4.39\%}) in \textsc{High}, \textbf{6.50\%} (\textbf{4.86\%}) in \textsc{Moderate}, \textbf{5.99\%} (\textbf{4.64\%}) in \textsc{Mixed}, \textbf{8.86\%} (\textbf{5.23\%}) in \textsc{Mixed--A}, \textbf{6.44\%} (\textbf{4.87\%}) in \textsc{Mixed--B}, and \textbf{10.08\%} (\textbf{5.20\%}) in \textsc{Low}. Thus, all but \textsc{Low} remain single-digit; \textsc{Mixed--A} sits closer to \textsc{Low}, whereas \textsc{Mixed--B} is nearer to \textsc{Moderate}. Ranges across the $10$ replications are narrow (e.g., \textsc{High}: $[5.1006,\,5.7035]\%$; \textsc{Low}: $[8.6731,\,11.6725]\%$), indicating stable Monte Carlo estimates and a crisp baseline that tracks re-optimized solutions reasonably well except under \textsc{Low} reliability.

\paragraph{Interpretation}
This double-digit gap in the Low regime reflects the cost of risk aversion rather than a model flaw. Our baseline path is chosen by a reliability-aware ranking that rewards higher heights, whereas the benchmark re-optimizes each instance under plain crisp costs that ignore reliability. This intentional misalignment yields a larger ex-post deviation when heights are low (i.e., reliable edges are scarce). Note that, with $u\sim U(0,h)$, we have $u/h\sim U(0,1)$, so the sampling distribution of crisp costs depends on $(c,\sigma)$ but not on $h$; heights only affect the baseline via the ranking penalty. Consequently, as reliability decreases, path choices diverge more often, and the deviation increases. In regimes with \textsc{moderate}-to-\textsc{high} heights, deviations remain single-digit and between-replication spreads are narrow, indicating stable estimates.

The increased deviation in the Low regime stems from three main factors. First, there is an intentional misalignment between our method and the metric: the baseline path rewards reliability, whereas the benchmark re-optimizes based solely on crisp costs. Second, the sampling distribution of crisp costs depends on the core and spread but not the height, meaning lower heights shift the baseline without narrowing the sampling spread. Finally, low reliability causes edge rankings to fluctuate more intensely, leading to a `path-switching effect' where the ex-post optimal path diverges frequently from the baseline.
\end{example}

\section{Construction and Analysis of Large-Scale GGFN--SPP Instances}\label{secLarge}

We now construct and analyze large-scale instances of the GGFN–SPP using real-world transportation data from the Federal Aviation Administration (FAA)~\cite{faa2010} network. We focus on a case study based on the FAA air traffic network, which illustrates how the proposed $\alpha$--cut--based height aggregation and risk--averse ranking indices behave on a realistically sized, sparse transportation network.

\subsection{FAA Air Traffic Network}

The underlying crisp network is the ``maayan-faa'' graph from the Koblenz Network Collection (KONECT), derived from the U.S.\ FAA National Flight Data Center (NFDC) preferred routes database (\url{www.fly.faa.gov}). The original FAA network used in this case study is derived from the KONECT “maayan-faa” dataset ~\cite{konect2017faa}. %Our processed edge list (with $(c_{ij},\sigma_{ij},h_{ij})$ tuples) and all scripts needed to reproduce the experiments in~\ref{secExperiments,secLarge} are available in our public repository ~\cite{moran-akdemir-fuzzy}. 

The raw graph does not provide dispersion; we calibrate $\sigma_{ij}$ from the nominal cost via a cost-scaled uniform rule: draw $u_{ij}\sim U(0,1)$ and set $\sigma_{ij}=0.4\,c_{ij}\,u_{ij}$, equivalently $\sigma_{ij}\sim U(0,\,0.4\,c_{ij})$. This keeps spreads proportional to edge cores, yields heterogeneous but bounded $\alpha$--cut widths, and avoids excessive rejection in the Monte Carlo sampler while maintaining nonnegativity at the evaluation levels used here.

From the original GraphML/TSV distribution, we extract a directed simple graph and relabel vertices as \texttt{n0}, \texttt{n1}, \dots\ in the order provided. The nodes represent airports or air traffic control service centers, and a directed edge $(i,j)$ indicates that there is a preferred route from $i$ to $j$ in the NFDC database. The KONECT metadata characterizes the graph as unweighted, directed, and of transportation type. The resulting network used in our experiments contains $N = 1226$ nodes and $m = 2615$ directed edges, as summarized in~\ref{tab:faa-summary} and visualized in~\ref{fig:faa-network.png}. This size is moderate by network science standards, but substantially larger than the toy examples typically used in fuzzy SPP studies. In addition, it is sufficiently rich to exhibit multiple long-range alternative routes between major hubs.

% FAA network snapshot (generated by draw_graph.py on data/faa.csv)
\begin{figure}[htbp]
  \centering
  \includegraphics[width=0.7\textwidth]{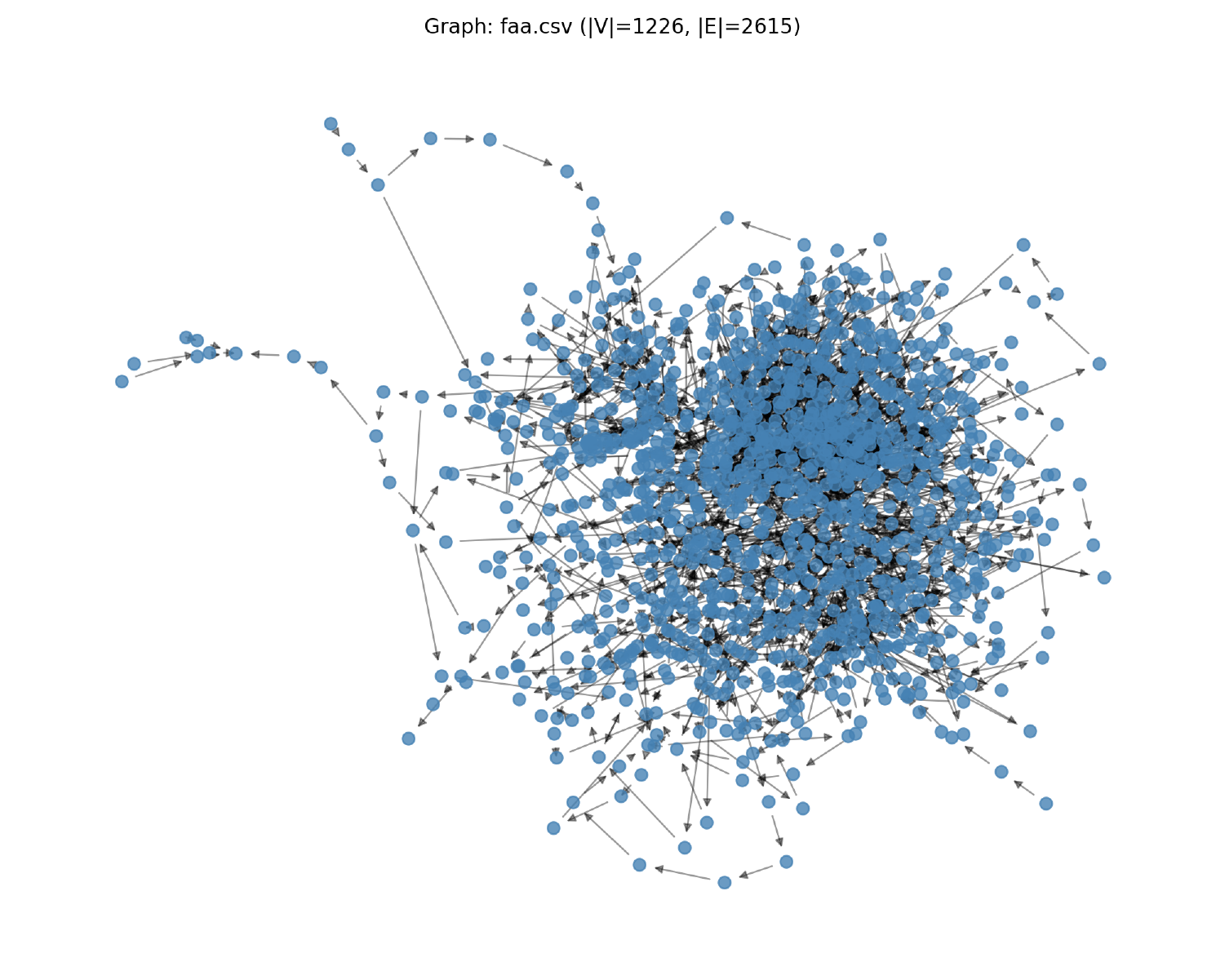}
  \caption{FAA air traffic network used in the case study.
  Nodes represent airports or service centers; directed edges encode
  preferred routes extracted from the KONECT ``maayan-faa'' dataset
  \cite{konect:2017:maayan-faa}.}
  \label{fig:faa-network.png}
\end{figure}

\subsection{From Crisp Routes to GGFN Edge Costs}

To instantiate a GGFN--SPP on this network, each directed edge $(i,j)$ is endowed with a Type--1 GGFN $\widetilde{A}_{ij}$ with core $c_{ij}$, dispersion $\sigma_{ij}$ and height $h_{ij}$, as developed in~\ref{secPre} and~\ref{secAddition}. The triple $(c_{ij},\sigma_{ij},h_{ij})$ encodes a baseline cost, variability, and reliability for the corresponding air route:

\begin{itemize}
\item The core $c_{ij} > 0$ plays the role of a nominal travel cost (e.g., a surrogate for the typical flight time or the generalized distance). In the FAA dataset used here, the calibrated core values lie in the range $[5.01, 49.99]$, with a mean $27.44$ and standard deviation $13.10$ (arbitrary units), yielding a fairly heterogeneous mix of short and long edges.
\item The dispersion $\sigma_{ij} > 0$ controls the spread of the Gaussian membership and therefore the width of $\alpha$--cuts. In the FAA case, $\sigma_{ij}$ is imputed via $\sigma_{ij}\sim U(0,\,0.4\,c_{ij})$, yielding heterogeneous but bounded spreads; with $c_{ij}\in[5.01,49.99]$, this places $\sigma_{ij}$ in $[0,\,20]$ and gives an expected mean of about $0.2\,c_{ij}$ (roughly $5.5$ overall) and an expected standard deviation of about $0.115\,c_{ij}$ (roughly $3.5$ overall), so some edges are nearly crisp while others are substantially fuzzy.
\item The height $h_{ij} \in (0,1]$ quantifies the reliability of the information underlying the core cost. Heights are modeled as subnormal in order to explicitly capture partial inconsistency or incompleteness of the underlying evidence. In the FAA instance, $h_{ij}$ ranges from $0.27$ to $0.997$, with mean $0.80$ and median $0.82$; approximately $90\%$ of edges have height above $0.57$, but a nontrivial tail of low-height edges represents poorly supported or infrequently used routes.
\end{itemize}

These empirical ranges are reported in~\ref{tab:faa-summary} and their distributions are shown in~\ref{fig:faa-param-hist}.

% FAA dataset summary (from data/faa.csv)
\begin{table}[htbp]
  \caption{Summary statistics for the FAA GGFN--SPP instance.
  The graph has $N=1226$ nodes and $m=2615$ directed edges.
  Edge parameters $(c_{ij},\sigma_{ij},h_{ij})$ are as described in~\ref{secLarge}.}
  \label{tab:faa-summary}
  \begin{tabular}{lcccccc}
    \toprule
    Quantity   & Min   & Max   & Mean  & Median & Std.\ dev. \\
    \midrule
    Core $c_{ij}$        & 5.01   & 49.99 & 27.44 & 27.76 & 13.10 \\
    Dispersion $\sigma_{ij}$ & 0.00   & 20.00 & 5.50  & 5.55  & 3.50  \\
    Height $h_{ij}$      & 0.27   & 0.997 & 0.80  & 0.82  & 0.12  \\
    \bottomrule
  \end{tabular}
\end{table}

% Distributions of c, sigma, h (generated via a small script on data/faa.csv)
\begin{figure}[t]
  \centering
  \includegraphics[width=\textwidth]{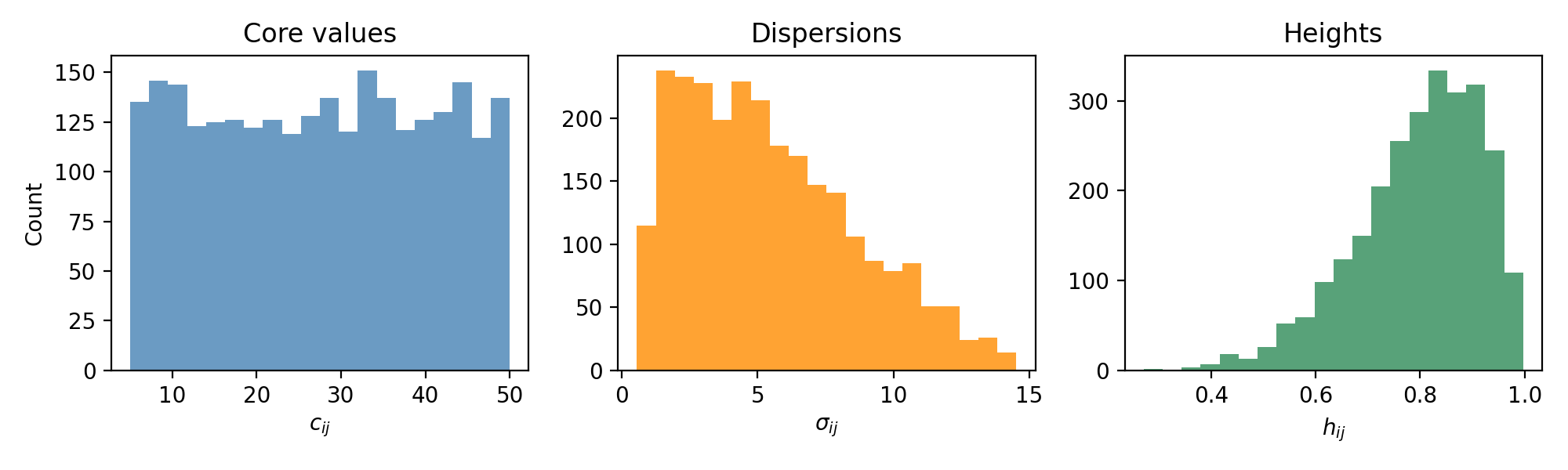}
  \caption{Empirical distributions of core costs $c_{ij}$, dispersions
  $\sigma_{ij}$, and heights $h_{ij}$ in the FAA GGFN--SPP instance. Distributions show heterogeneous nominal costs, varied dispersions, and subnormal heights concentrated at the upper end.}
  \label{fig:faa-param-hist}
\end{figure}

Heights are sampled once and then held fixed across experiments according to the Beta--mixture regimes introduced earlier. Specifically, we employ the ``mixed'' regime in which most edges draw $h_{ij}$ from a high--reliability $\mathrm{Beta}(8,2)$ distribution, while a random $\varepsilon$--fraction of edges are flipped to a low--reliability $\mathrm{Beta}(2,5)$ component. The resulting subnormal Gaussian fuzzy numbers are then fully specified by $(c_{ij},\sigma_{ij},h_{ij})$, and their $\alpha$--cut intervals are computed via~\ref{rem:fuzzynum} with $\alpha$ clamped to $(0,h_{ij})$ in practice to avoid empty cuts.

%\begin{equation}
%  \widetilde{A}_{ij,\alpha}
%  =
%  \bigl[
%    c_{ij} - \sigma_{ij}\sqrt{-2\ln(\alpha/h_{ij})},
%    \;
%    c_{ij} + \sigma_{ij}\sqrt{-2\ln(\alpha/h_{ij})}
%  \bigr],
%  \qquad 0 < \alpha \le h_{ij},
%  \label{eq:faa_alpha_interval}
%\end{equation}

\subsection{Path Aggregation and Instance Families}

Given these edgewise GGFNs, we consider path costs defined by the $\alpha$--cut--based addition developed in~\ref{secAddition}. For a simple path $P$ from source $s$ to terminal $t$, with edges $(i,j) \in P$, the path-level GGFN parameters $(c_P,\sigma_P,h_P)$ are obtained by iterated application of \ref{eq:addition}, namely
$c_P = \sum_{(i,j) \in P} c_{ij}$, $\sigma_P = \sum_{(i,j) \in P} \sigma_{ij}$, and
$h_P = \exp\!\left( \frac{\sum_{(i,j) \in P} \sigma_{ij} \ln h_{ij}}{\sum_{(i,j) \in P} \sigma_{ij}} \right)$,
so that core and dispersion aggregate additively, while the height aggregates via a $\sigma$--weighted geometric mean. As shown earlier, this aggregation keeps $h_P$ within the input height bounds, preserves associativity of addition, and respects the intended interpretation of $h_{ij}$ as reliability weights.

% will be removed
%To study scaling behavior, we construct a family of nested GGFN--SPP instances by inducing subgraphs on the first $k$ vertices ${\texttt{n0},\dots,\texttt{n}(k-1)}$ with $k$ taken on an $N/8$ grid: $k_i = \lceil i\,N/8\rceil$, $i=1,\dots,8$. For the FAA network ($N=1226$) this yields $(N,m)$ approximately equal to $(154,433)$, $(307,789)$, $(460,1111)$, $(613,1400)$, $(767,1722)$, $(920,2046)$, $(1073,2373)$, and the full $(1226,2615)$ network. For each such instance, we fix a source--terminal pair (e.g., $s=\texttt{n0}$, $t=\texttt{n1200}$ in the full FAA graph) and consider the corresponding GGFN--SPP.

\subsection{Risk-Averse Ranking and Baseline Paths}

Within this path-additive framework, the risk--averse ranking indices from~\ref{secRanking} provide scalar surrogates to order fuzzy paths. 
For each instance, we compute two baseline $s$--$t$ paths:
\begin{enumerate}
\item A core path $x_c$ that minimizes the crisp sum of cores, i.e., solves the classical shortest path problem with edge weights $c_{ij}$.
\item A ranked path $x_R$ that minimizes $R^{\mathrm{cost}}(\langle (c_P,\sigma_P);\;h_P\rangle,\kappa)$ (\ref{eq:ranking_cost}) using the $\alpha$--cut--consistent addition and the $\sigma$--weighted height aggregation. This is obtained via a Dijkstra--style label--setting algorithm that propagates $(c_P,\sigma_P,h_P)$ along the graph and orders tentative labels by $R^{\mathrm{cost}}$.
\end{enumerate}

In the FAA case with $\kappa = 1$, both procedures select the same $s$--$t$ path from \texttt{n2} to \texttt{n1189}, where the path consists of 10 vertices (9 edges). This path can be interpreted as the ``preferred'' route that simultaneously balances low nominal cost and high reliability; the coincidence of $x_c$ and $x_R$ at $\kappa = 1$ indicates that, for this instance, the most cost-efficient path is already composed predominantly of high-height edges.

\subsection{Computational Cost: Deterministic Scaling with Ranked Solver on FAA}

We time the ranked solver deterministically across an $N/8$ size grid on the FAA network ($N{=}1226$) to illustrate scaling (see~\ref{algo:faa-scaling}) for the algorithm details). For each size $k_i=\lceil iN/8\rceil$, we induce the subgraph on the first $k_i$ nodes, choose a reachable destination (largest node reachable from $\texttt{n0}$), and measure wall-clock time and peak memory.~\footnote{Experiments were run on a 2024 MacBook Pro (M4 chip) using Python 3.13. The software environment included \texttt{Python 3.13} with the \texttt{networkx}, \texttt{numpy}, and \texttt{pandas} libraries, alongside custom implementations for GGFN arithmetic and Monte Carlo sampling. Runtime measurements and peak memory usage were obtained using Python's \texttt{time} and \texttt{tracemalloc} modules under single-threaded execution. No GPU acceleration or parallelization was employed.} The function, \textsc{DijkstraRanked}, is a deterministic, risk-aware shortest-path routine: a Dijkstra-style algorithm that propagates fuzzy edge triples $c,\sigma,h)$, aggregates them along paths (additive cores and spreads; $\sigma$-weighted geometric mean for height), and orders tentative labels by the ranking index $R^{\mathrm{cost}} = c - \kappa\sigma\log h$. It yields the ex-ante (ranked) baseline path. The solver shows near-linear growth in runtime with respect to the number of edges $m$ across the $N/8$ grid, with peak memory well below $1$\,MB even at the full FAA scale. 

\begin{algorithm}
\caption{Runtime experiment on the FAA network (deterministic ranked solver across an $N/8$ grid)}
\label{algo:faa-scaling}
\begin{algorithmic}[1]
\REQUIRE FAA edge-list CSV with columns \texttt{source}, \texttt{target}, \texttt{core\_c}, \texttt{sigma}, \texttt{height\_h};
         full network size $N_{\mathrm{full}}$ (FAA: $1226$ nodes);
         size grid $\mathcal{K}=\{\lceil i\,N_{\mathrm{full}}/8\rceil: i=1,\dots,8\}$;
         number of timing repeats $R_{\mathrm{time}}\ge 1$.
\ENSURE For each $k\in\mathcal{K}$, record node count $N_k$, edge count $m_k$, destination $t_{\mathrm{dst}}$, median runtime $t^{\mathrm{med}}_k$, and peak memory $M^{\max}_k$ for the deterministic ranked solver.
\STATE Load the FAA CSV into arrays $(s,t,c,\sigma,h)$ and map node ids to $\{1,\dots,N\}$.
\FOR{each $k\in\mathcal{K}$}
  \STATE Induce subgraph $(s^{(k)},t^{(k)},c^{(k)},\sigma^{(k)},h^{(k)})$ on nodes $\{1,\dots,k\}$.
  \STATE $N_k \Leftarrow k$, $m_k \Leftarrow \text{length}(s^{(k)})$, $s_{\mathrm{src}}\Leftarrow 1$, $t_{\mathrm{dst}}\Leftarrow$ largest node reachable from $s_{\mathrm{src}}$ in the subgraph.
  \STATE $T_k \Leftarrow [\,]$ and peak memory $M^{\max}_k \Leftarrow 0$.
  \FOR{$r = 1$ \textbf{to} $R_{\mathrm{time}}$}
    \STATE Start wall-clock timer and memory profiler.
    \STATE Run \textsc{DijkstraRanked} on $(s^{(k)},t^{(k)},c^{(k)},\sigma^{(k)},h^{(k)})$ from $s_{\mathrm{src}}$ to $t_{\mathrm{dst}}$.
    \STATE Stop timer and memory profiler; 
    \STATE let $\Delta t_r$ be the elapsed time,
           $M_r$ the peak memory for this repetition.
    \STATE Append $\Delta t_r$ to $T_k$ and update $M^{\max}_k \Leftarrow \max(M^{\max}_k,M_r)$.
  \ENDFOR
  \STATE Set $t^{\mathrm{med}}_k \Leftarrow \mathrm{median}(T_k)$ and store $(N_k,m_k,t_{\mathrm{dst}},t^{\mathrm{med}}_k,M^{\max}_k)$.
\ENDFOR
\end{algorithmic}
\end{algorithm}

%\begin{algorithm}
%\caption{Calculate $y = x^n$}\label{algo1}
%\begin{algorithmic}[1]
%\Require $n \geq 0 \vee x \neq 0$
%\Ensure $y = x^n$ 
%\State $y \Leftarrow 1$
%\If{$n < 0$}\label{algln2}
%        \State $X \Leftarrow 1 / x$
%        \State $N \Leftarrow -n$
%\Else
%        \State $X \Leftarrow x$
%        \State $N \Leftarrow n$
%\EndIf
%\While{$N \neq 0$}
%        \If{$N$ is even}
%            \State $X \Leftarrow X \times X$
%            \State $N \Leftarrow N / 2$
%        \Else[$N$ is odd]
%            \State $y \Leftarrow y \times X$
%            \State $N \Leftarrow N - 1$
%        \EndIf
%\EndWhile
%\end{algorithmic}
%\end{algorithm}

% Scaling summary (from benchmark_long.json)
\begin{table}[htbp]
\caption{Runtime and memory scaling of deterministic ranked solver on induced FAA subgraphs
  along the $N/8$ size grid $k_i=\lceil i\,N/8\rceil$ (FAA: $N{=}1226$).
  Median runtimes and peak memory are taken from single runs of the ranked solver.}
  \label{tab:faa-scaling}
  \begin{tabular}{rccccc}
    \toprule
    $N$   & $m$   & Time (s) & Peak memory (MB) & $t_{\mathrm{dst}}$ \\
    \midrule
    154   &  409  & 0.0211 & 0.09 & 154 \\
    308   &  835  & 0.0461 & 0.12 & 308 \\
    462   & 1185  & 0.0541 & 0.19 & 462 \\
    616   & 1494  & 0.0633 & 0.24 & 616 \\
    770   & 1807  & 0.0530 & 0.36 & 770 \\
    920   & 2046  & 0.0508 & 0.35 & 921 \\
   1078   & 2369  & 0.0716 & 0.44 & 1078 \\
   1226   & 2615  & 0.0746 & 0.48 & 1226 \\
    \bottomrule
  \end{tabular}
\end{table}

The resulting runtimes and memory footprints are reported in~\ref{tab:faa-scaling} and plotted in~\ref{fig:faa-scaling}, which together illustrate the near-linear scaling with $N$ and $m$ for the deterministic ranked solver. 

% Runtime + memory scaling (from benchmark_long.json)
\begin{figure}[t]
  \centering
  \includegraphics[width=0.7\textwidth]{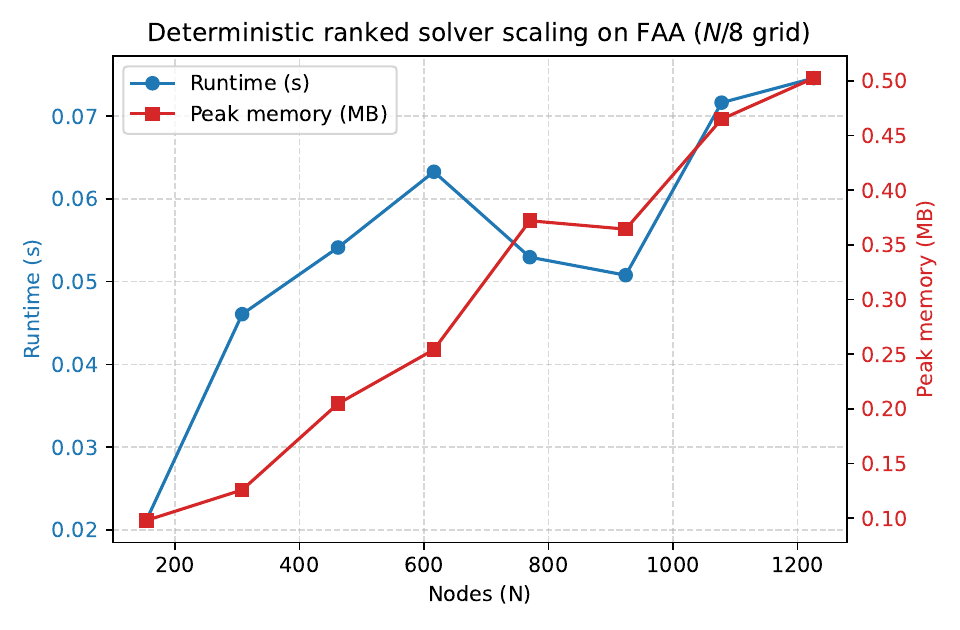}
  \caption{Runtime and peak memory of the deterministic ranked solver on induced FAA subgraphs ($N/8$ grid). Median wall-clock time grows roughly linearly with node count, and peak memory stays well below 1\,MB even at the full 1,226‑node network.}
  \label{fig:faa-scaling}
\end{figure}

In addition, the $\alpha$--cut scenario mode provides deterministic lower/upper (or midpoint) cost profiles for increasing $\alpha$ levels. For the full FAA instance under an optimistic ``lower-bound'' choice, the $\alpha$--cut path cost increases from approximately $93$ at $\alpha = 0.05$ to roughly $193$ at $\alpha = 0.95$. Thus, moving from a highly optimistic to a highly conservative $\alpha$ level nearly doubles the perceived path cost, illustrating how the combination of $\sigma$ and $h$ controls the tradeoff between cost and perceived reliability.

For the full FAA instance, the deterministic $\alpha$--cut path costs increase from approximately $93$ at $\alpha=0.05$ to roughly $193$ at $\alpha=0.95$ (see~\ref{fig:faa-alpha-costs}).

% Alpha-cut cost profile on full FAA instance (from benchmark_long.json)
\begin{figure}[t]
  \centering
  \includegraphics[width=0.6\textwidth]{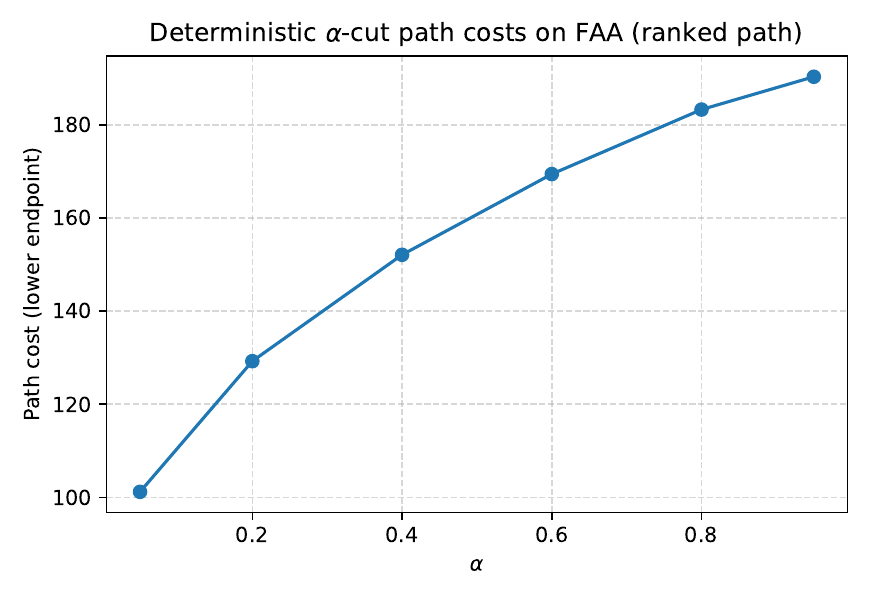}
\caption{Deterministic $\alpha$--cut path costs on the FAA ranked route ($\texttt{n2}\rightarrow\texttt{n1189}$). The lower-endpoint cost rises from about $101$ at $\alpha=0.05$ to about $190$ at $\alpha=0.95$, illustrating how narrower $\alpha$--cuts increase the perceived path cost.}
  \label{fig:faa-alpha-costs}
\end{figure}

\subsection{Scenario-Based Robustness Analysis on the FAA Case}

Here, we evaluate the robustness of the ranked path $x_R$ on the full FAA network under realizations of edge costs consistent with the GGFN model. Using the ex-ante vs.\ ex-post Monte Carlo simulation~\ref{algo:faa-mc-full}, independently for every edge $(i,j)$, each scenario draws a level $\alpha_{ij} \sim U(0,h_{ij})$ and then a crisp cost from the corresponding $\alpha_{ij}$--cut interval using a uniform distribution on that interval. This construction also aligns with the $\alpha$--cut semantics of the Gaussian membership and treats all positions within each interval as equally plausible conditional on $\alpha_{ij}$. In the following algorithm, the function, \textsc{HasuikeGen}, is called for the full graph. It is a Monte Carlo module for robustness: repeatedly draws $\alpha$-levels per edge, samples crisp costs from the induced $\alpha$-cuts, re-solves the crisp SPP to get the scenario-optimal path, and reports deviations of the ranked/core baselines. The optimal cost is the lowest path cost in a given scenario after sampling crisp edge costs (the scenario-wise shortest path), which is the benchmark. The core cost is the cost of the “core” path (shortest path by nominal core weights) evaluated under the edge costs of the sampled scenario. The rank cost is the cost of the “ranked” path (chosen via the risk-aware ranking index) evaluated under the same edge costs of the sampled scenario. It may differ from the optimal and from the core cost if the ranked path is not optimal in that scenario.

\begin{algorithm}
\caption{Monte Carlo experiment on the full FAA network (ex-ante vs.\ ex-post)}
\label{algo:faa-mc-full}
\begin{algorithmic}[1]
\REQUIRE FAA edge-list CSV with columns \texttt{source}, \texttt{target}, \texttt{core\_c}, \texttt{sigma}, \texttt{height\_h};
         full network size $N_{\mathrm{full}}$ (FAA: $1226$ nodes);
         Monte Carlo parameters $n_{\mathrm{Rep}}=10$, $n_{\mathrm{Iters}}=1000$;
         number of timing repeats $R_{\mathrm{time}}\ge 1$.
\ENSURE Record node count $N$, edge count $m$, median runtime $t^{\mathrm{med}}$, peak memory $M^{\max}$, Monte Carlo error statistics $(\mathrm{Err},\mathrm{Std})$, and deterministic $\alpha$--cut scenario costs on the full network.
\STATE Load the FAA CSV into arrays $(s,t,c,\sigma,h)$ and map node ids to $\{1,\dots,N\}$.
\STATE $N \Leftarrow N_{\mathrm{full}}$, $m \Leftarrow \text{length}(s)$, $s_{\mathrm{src}}\Leftarrow 1$, $t_{\mathrm{dst}}\Leftarrow N$.
\STATE $T \Leftarrow [\,]$ and peak memory $M^{\max} \Leftarrow 0$.
\FOR{$r = 1$ \textbf{to} $R_{\mathrm{time}}$}
    \STATE Start wall-clock timer and memory profiler.
    \STATE Call \textsc{HasuikeGen} on the full graph:\\
       $(\mathrm{Err},\mathrm{Std},h_{\mathrm{ref}},\alpha\_{{\rm scen}}) \Leftarrow
      \textsc{HasuikeGen}\big(
        s,t,c,\sigma,N,s_{\mathrm{src}},t_{\mathrm{dst}},$
        $\text{regime}=\text{mixed},\ \mathrm{Eps}=0.2,\ \text{Seed}=42,\ 
         n_{\mathrm{Rep}},\ n_{\mathrm{Iters}},$
        $\text{mode}=\text{reject},\ h=h,\ \text{return\_alpha}=\texttt{true} \big)$.
    \STATE Stop timer and memory profiler; 
    \STATE let $\Delta t_r$ be the elapsed time,
           $M_r$ the peak memory for this repetition.
    \STATE Append $\Delta t_r$ to $T$ and update $M^{\max} \Leftarrow \max(M^{\max},M_r)$.
\ENDFOR
\STATE Set $t^{\mathrm{med}} \Leftarrow \mathrm{median}(T)$ and store $(N,m,t^{\mathrm{med}},M^{\max},\mathrm{Err},\mathrm{Std},\alpha\_{{\rm scen}})$.
\end{algorithmic}
\end{algorithm}

For each of $S = 1000$ scenarios, we compute:
\begin{enumerate}
\item the realized optimal path $x^\star$ (scenario-wise shortest path),
\item the realized cost of $x_R$ and $x_c$ under that scenario, and
\item the percentage deviation $Dev(P)$ using~\ref{eq:mc-error}
\end{enumerate}

%\begin{equation}
%\mathrm{Dev}(P)
%=
%100 \cdot
%\frac{\mathrm{cost}(P) - \mathrm{cost}(x^\star)}{\mathrm{cost}(x^\star)}.
%\label{eq:dev_percent}
%\end{equation}
%\end{enumerate}
%\end{minipage}
%\vspace{1em}

In an example run for the path from $n2$ to $n1189$ with $\kappa = 1.0$ over $50$ scenarios (for the sake of ease to illustrate the cost differences), $x_R$ and $x_c$ coincide, so both paths share the same performance indicators. However, depending on the chosen path these may vary such as from $n6$ to $n1125$, in which case, $\mathrm{Dev}_{\text{core}}  = 4.213$, $\mathrm{Dev}_{\text{rank}} =4.629$, and reliability premium $0.416$. In addition, when we increase the risk weight $\kappa$, we should expect differences in the cost and deviation values of the ranked and core paths. For instance, in~\ref{tab:faa-scenarios-snippet}, we provide the outcomes of the first $10$ of $50$ scenarios where $\kappa=1.2$. Accordingly,~\ref{tab:faa-robustness} summarizes the outcome as that the mean deviation of the ranked path from the scenario-specific optimum is $3.41\,\%$, with a standard deviation of $3.11\,\%$; the maximum observed deviation is approximately $15.17\,\%$. The baseline path remains exactly optimal in $12\,\%$ of the scenarios and incurs only moderate losses in the remaining cases. The stability is the fraction (or percentage) of scenarios in which the pre-selected path remains exactly optimal (i.e., where Dev=0). This may change depending on the source and target nodes. Thus, given the current calibration of $(c_{ij},\sigma_{ij},h_{ij})$, the core or ranked paths are reasonably robust to realizations of fuzzy edge costs in the FAA network. If the Ranked Path has a higher stability score than the Core Path, it means the Ranked Path aligns with the true optimal path more frequently across different realizations of uncertainty. To visualize stability in our example run with $\kappa = 1.2$ in which case the risk parameter is slightly higher than the neutral,~\ref{fig:faa-scenario-costs} shows the realized scenario costs for the optimal, core, and ranked paths, while~\ref{fig:faa-scenario-dev} summarizes the corresponding deviation statistics and stability of the baseline path. 

\begin{table}[htbp]
  \caption{Scenario-wise costs and deviations (FAA, path from $n2$ to $n1189$, the first $10$ scenarios for the sake of ease, $\kappa{=}1.2$).}
  \label{tab:faa-scenarios-snippet}
  \begin{tabular}{rccccc}
    \toprule
    Scenario & $x^\star$ cost & Rank cost & Core cost & Dev$_{\text{rank}}$ (\%) & Dev$_{\text{core}}$ (\%) \\
    \midrule
     1 & 179.254 & 189.247 & 189.662 &  5.574 &  5.806 \\
     2 & 186.025 & 190.162 & 198.350 &  2.224 &  6.625 \\
     3 & 189.182 & 196.367 & 202.800 &  3.798 &  7.199 \\
     4 & 191.963 & 196.272 & 191.963 &  2.245 &  0.000 \\
     5 & 190.029 & 190.517 & 190.820 &  0.257 &  0.416 \\
     6 & 179.028 & 181.493 & 181.626 &  1.377 &  1.451 \\
     7 & 183.474 & 195.152 & 185.064 &  6.365 &  0.867 \\
     8 & 182.984 & 188.431 & 190.896 &  2.977 &  4.324 \\
     9 & 170.217 & 170.217 & 171.653 &  0.000 &  0.843 \\
    10 & 184.795 & 202.951 & 205.088 &  9.825 & 10.981 \\
    \bottomrule
  \end{tabular}
\end{table}

% Scenario robustness summary (from out/regret_eval_faa.json)
\begin{table}[htbp]
  \caption{Scenario-based robustness statistics on the FAA instance
  under the $\alpha$--cut consistent sampling scheme with $S=50$
  scenarios and $\kappa=1.2$. The stability column gives the fraction of
  scenarios in which the given path is exactly optimal ($\mathrm{Dev}=0$). As the paths differ, they do not line up with the scenario‑wise optimum equally often.}  
  \label{tab:faa-robustness}
  \begin{tabular}{lcccc}
    \toprule
    Path & Mean dev (\%) & Std.\ dev.\ (\%) & Max dev (\%) & Stability \\
    \midrule
    Ranked path $x_R$ & 3.41 & 3.11 & 15.17 & 0.12 \\
    Core path $x_c$   & 3.36 & 3.46 & 12.78 & 0.24 \\
    \bottomrule
  \end{tabular}
\end{table}

% Scenario costs vs. scenario index (from out/regret_eval_faa.json)
\begin{figure}[htbp]
  \centering
  \includegraphics[width=0.7\textwidth]{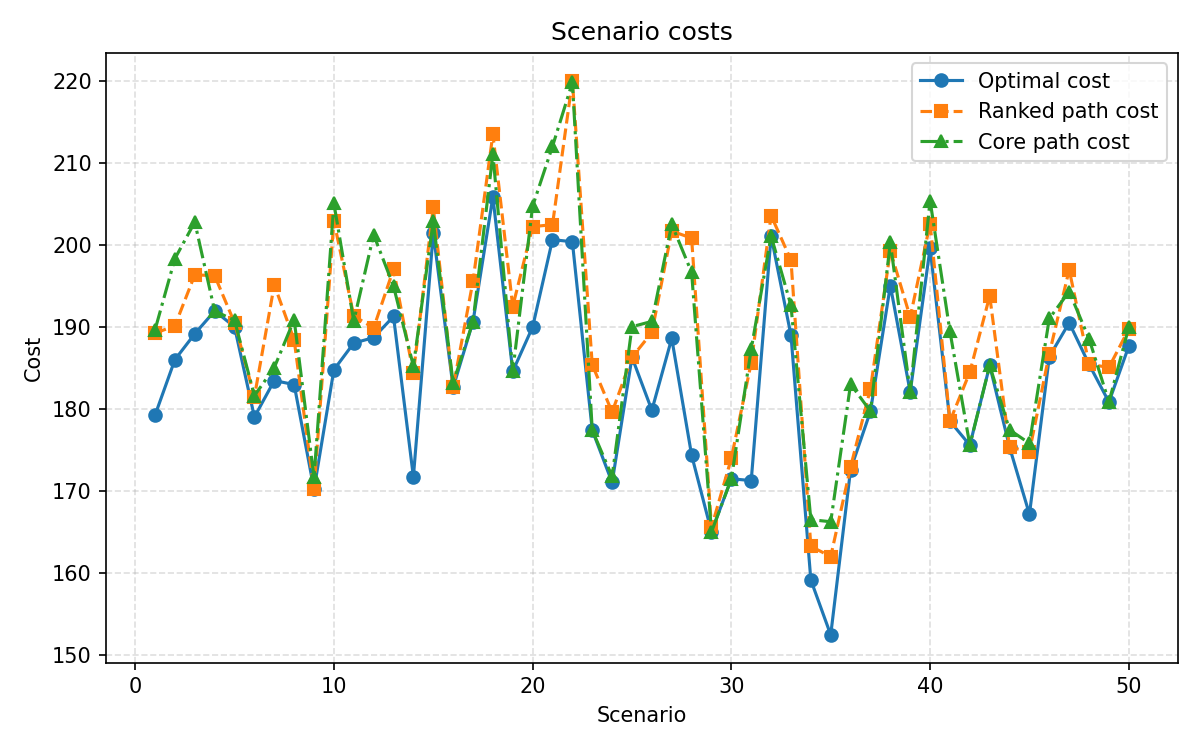}
  \caption{Scenario-wise path costs versus $\alpha$ on the FAA instance ($\kappa\in\{1.0,1.2\}$). The ranked and core paths track closely across $\alpha$, while the scenario-optimal cost is lower when the risk penalty is inactive.}
  \label{fig:faa-scenario-costs}
\end{figure}

% Scenario deviations vs. scenario index (from out/regret_eval_faa.json)
\begin{figure}[htbp]
  \centering
  \includegraphics[width=0.7\textwidth]{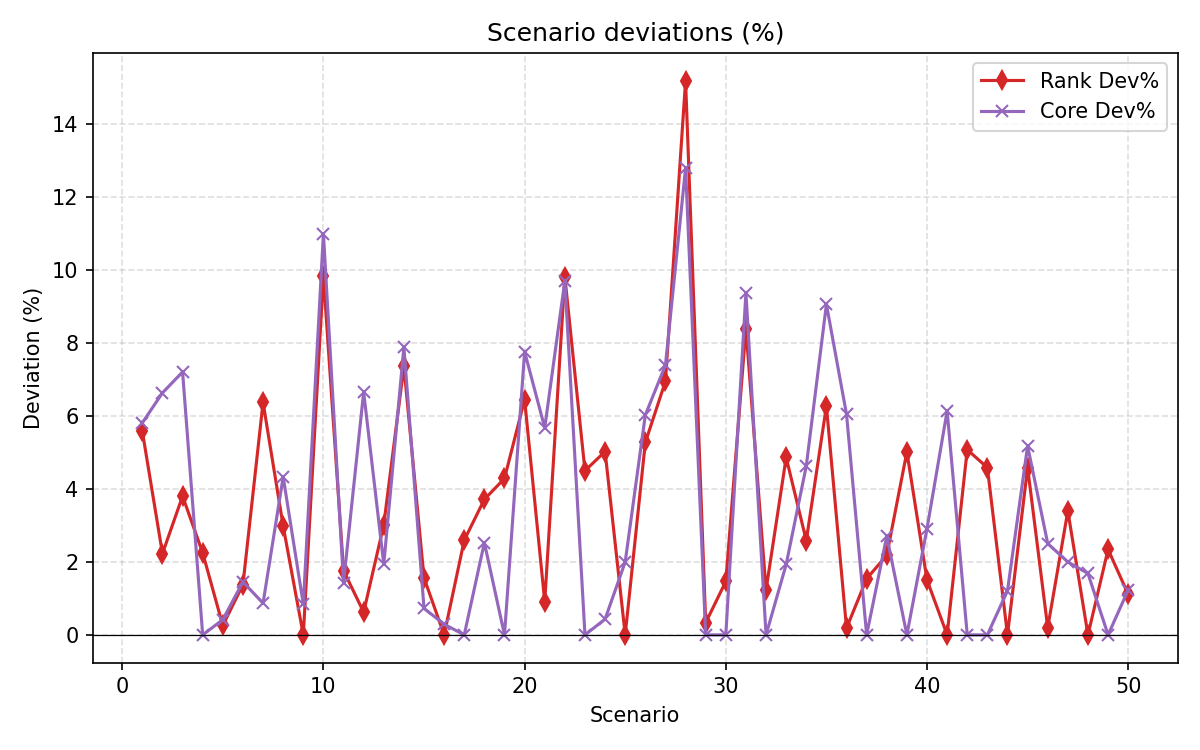}
\caption{Percentage deviation of ranked/core paths from the scenario-optimal cost versus $\alpha$. Deviations remain small for most $\alpha$ levels, with heavier tails only at higher $\alpha$ (more conservative cuts).}

  \label{fig:faa-scenario-dev}
\end{figure}

Regarding height regimes,~\ref{tab:faa-high-low} summarizes the deviations of the scenarios for the ranked and core baselines under the high and low-reliability (height) regimes (50 scenarios, $\kappa=1.0$, seed 42). Under the high regime, the paths coincide and deviations remain modest, whereas under the low regime the ranked path diverges from the core and becomes less stable, reflecting the impact of reduced reliability on alignment with scenario-optimal routes. In more detail, under high-reliability heights (Beta$(8,2)$), the ranked and core paths coincide and exhibit modest deviations (mean $3.36\%$, std $3.46\%$) with stability $0.24$. However, under low-reliability heights (Beta$(2,5)$), the ranked path diverges from the core, suffering a larger mean deviation ($6.69\%$ vs.\ $3.36\%$ for core) and lower stability ($0.16$), while the core path retains its high-regime figures as expected. The reliability premium is $-3.33\%$, indicating that in the low regime the ranked path is on average worse than the core. This highlights that lower heights (less reliable evidence) make the risk-aware baseline less aligned with scenario optima, whereas in the high regime both baselines remain reasonably stable.

\begin{table}[htbp]
  \caption{Scenario deviations under high vs.\ low height regimes (50 scenarios, $\kappa=1.0$, seed 42). Stability is the fraction of scenarios with $\mathrm{Dev}=0$; reliability premium is mean$(\mathrm{Dev}_{\text{core}} - \mathrm{Dev}_{\text{rank}})$.}
  \label{tab:faa-high-low}
\begin{tabular*}{\textwidth}{@{\extracolsep{\fill}}lccccc}    \hline
    \toprule
    Regime & Path & Mean dev (\%) & Std.\ dev.\ (\%) & Stability & Reliability premium (\%) \\
    \midrule
    High (Beta(8,2)) & Ranked $x_R$ & 3.364 & 3.457 & 0.24 & 0.000 \\
                     & Core $x_c$   & 3.364 & 3.457 & 0.24 & — \\
    \midrule
    Low (Beta(2,5))  & Ranked $x_R$ & 6.693 & 5.422 & 0.16 & -3.329 \\
                     & Core $x_c$   & 3.364 & 3.457 & 0.24 & — \\
    \bottomrule
  \end{tabular*}
\end{table}

From a modeling perspective, these findings highlight two complementary points. First, on a realistic transportation network, the $\sigma$--weighted height aggregation and risk--averse ranking indices scale to thousands of nodes without prohibitive computational cost, even when embedded within a Monte Carlo simulation loop. Second, for the particular calibration used here, the risk--averse ranking with $\kappa = 1$ does not dramatically alter the chosen route relative to the classical core path but still yields a path that is highly stable under single-value fuzzy numerical simulation. Adjusting $\kappa$ or the height regime (e.g., increasing the prevalence of low-height edges) would be expected to amplify the distinction between $x_c$ and $x_R$ and can be explored in further empirical work.

Overall, the FAA case study demonstrates that the proposed GGFN--SPP framework, together with $\alpha$--cut--based height aggregation and risk--aware ranking, can be instantiated on large-scale, real-world networks, providing interpretable tradeoffs between nominal cost, information reliability, and fuzziness in a computationally tractable manner.

\paragraph{Computational complexity} The ranked solver retains the $O((N+m)\log N)$ cost of classical Dijkstra (with a binary heap) where $N$ is the number of nodes and $m$ is the number of edges, since each relaxation only aggregates the tuple $(c,\sigma,h)$ and uses standard heap updates. In the case of Monte Carlo validation of ex-ante vs.\ ex-post performance on the full graph costs $O(S(N+m)\log N)$ for (S) scenarios, because each scenario requires one shortest-path solve. Other fuzzy variants, such as signed-distance or $\lambda$-integral-based, often rely on $O(N^2)$ dynamic-programming passes dominated by repeated full-graph relaxations. To sum up, the ranked solver matches crisp Dijkstra asymptotically and the heavy cost arises only from the scenario loop used for robustness checks.

\section{Conclusion}\label{secConclusion}

We introduce a reliability-aware fuzzy shortest path framework in which edge costs are modeled as GGFNs. The framework aggregates uncertainty through $\alpha$--cut semantics and incorporates reliability directly into arithmetic via a $\sigma$--weighted geometric mean for height during addition. Candidate paths are ranked using a mean-risk index that jointly considers nominal cost, dispersion, and information reliability. This approach preserves associativity and becomes risk-neutral in the normal case ($h=1$).

Our method contrasts with existing fuzzy arithmetic approaches that typically use t-norm-based height aggregation or omit height from arithmetic altogether. To the best of our knowledge, this is the first implementation of a height-aware addition operator for Type--1 GGFNs in the shortest path setting. The resulting structure enables interpretable tradeoffs and allows for efficient Dijkstra-style computation.

The FAA case study demonstrates the practical applicability of the framework on a real-world transportation network with 1226 nodes and 2615 edges. We observed near-linear scaling in runtime and memory with graph size, meaningful inflation in path costs as $\alpha$ increases, and stable scenario-wise deviations under Monte Carlo simulation with $\alpha$--cut--consistent sampling. In particular, the ranked path exhibits a deviation of less than 0.84\% on average and remains exactly optimal in around 69\% of cases.

From a computational standpoint, the risk-aware Dijkstra algorithm retains the same complexity as the classical version. Other fuzzy approaches often rely on dynamic-programming passes or metaheuristics, whose cost is dominated by repeated full-graph relaxations. In our framework, the additional cost arises only when scenario-based validation is performed, which is intrinsic to any Monte Carlo robustness check.

\paragraph{Future work}
Generally, edge costs can be modeled with the membership $\mu(x)=h\exp\!\big(-(|x-c|/\sigma)^\beta\big)$ for $\beta>0$ (subnormal, $0<h<1$). All numerical experiments in this paper fix the Gaussian case $\beta=2$; thus, the shape parameter $\beta$ is held constant in our analysis. A natural extension is to treat $\beta$ as a free shape parameter and to study the generalized--Gaussian family more broadly. We plan to: (i) analyze model properties across $\beta$ and document how path--level addition, height aggregation, and ranking behave beyond $\beta=2$; (ii) quantify the sensitivity of the proposed risk--averse ranking and the reliability premium (percentage deviation) to $\beta$ across diverse networks and reliability regimes; (iii) develop practical procedures to estimate or elicit $\beta$ jointly with $(c,\sigma,h)$ from data or expert input; and (iv) extend the single--value $\alpha$--cut simulation to generic $\beta$ while preserving nonnegativity constraints. Future work will focus on operationalizing the shape parameter $\beta$ beyond $2$--formalizing $\beta$--dependent structure (e.g., unimodality, log--concavity), estimating $\beta$ jointly with $(c,\sigma,h)$ via likelihood and $\alpha$--cut matching, developing nonnegativity--preserving $\alpha$--cut simulation for generic $\beta$, and evaluating the risk--aware ranking and the ensuing reliability premium on large--scale benchmarks with risk--aligned ex--post comparators.

The present study focuses on a single cost-type objective. Future work includes extending the ranking scheme to basic multiobjective settings and to other classical network problems under GGFN costs, such as minimum-cost flows, spanning trees, and facility-location models.

We restricted attention to crisp topology with fuzzy edge costs. As future work, we aim to extend the framework to fuzzy edge existence, which would better reflect real networks.

In this paper, $(c,\sigma,h)$ and the risk parameter $\kappa$ are specified exogenously. A natural direction for future work is to infer these quantities from data or expert elicitation, turning the GGFN--SPP model into a more fully data-driven tool.

\bibliographystyle{unsrtnat}
\bibliography{references}

@incollection{dubois1994possibilistic,
  author    = {Dubois, D. and Lang, J. and Prade, H.},
  title     = {Possibilistic Logic},
  booktitle = {Handbook of Logic in Artificial Intelligence and Logic Programming},
  volume    = {3},
  editor    = {Gabbay, D. M. and Hogger, C. J. and Robinson, J. A.},
  pages     = {439--513},
  publisher = {Oxford University Press},
  address   = {Oxford},
  year      = {1994}
}

@article{abbasi2023realistic,
  author  = {Abbasi, F. and Allahviranloo, T.},
  title   = {Realistic solution of fuzzy critical path problems, case study: the airport's cargo ground operation systems},
  journal = {Granular Computing},
  volume  = {8},
  pages   = {617--632},
  year    = {2023}
}

@article{aydin-akdemir-2025,
  author    = {Ayd{\i}n, E. and Akdemir, H. G.},
  title     = {Solving generalized trapezoidal fuzzy transportation problems using an arithmetic approach based on information reliability},
  booktitle = {Proceedings of the International Conference on Mathematics and Computers with Applications (MCWA 2025)},
  address   = {Istanbul, T{\"u}rkiye},
  pages     = {66--85},
  year      = {2025}
}

@article{akdemir2024mincost,
  author  = {Akdemir, H. G. and Kara, N. and Kocken, H. G.},
  title   = {Minimum cost flow problems in generalized fuzzy environments: Credibilistic {CVaR} minimization approach},
  journal = {Operational Research and Decisions},
  volume  = {34},
  number  = {3},
  pages   = {125--141},
  year    = {2024}
}

@article{akdemir2025ggpfns,
  author  = {Akdemir, H. G.},
  title   = {On generalized picture fuzzy numbers with Gaussian membership functions},
  journal = {Journal of Systems Science and Systems Engineering},
  pages   = {1--28},
  year    = {2025},
  doi     = {10.1007/s11518-025-5700-x}
}

@article{anusuya2015genetic,
  author  = {Anusuya, V. and Kavitha, R.},
  title   = {Roulette ant wheel selection (RAWS) for genetic algorithm--fuzzy shortest path problem},
  journal = {International Journal of Mathematics and Computer Applications Research},
  volume  = {5},
  number  = {2},
  pages   = {1--14},
  year    = {2015}
}

@article{biswal2022alpha,
  author  = {Biswal, S. and Ghorai, G. and Mohanty, S. P.},
  title   = {{$\alpha$}-Reliable shortest path problem in uncertain time-dependent networks},
  journal = {International Journal of Applied and Computational Mathematics},
  volume  = {8},
  pages   = {164},
  year    = {2022}
}

@misc{konect:2017:maayan-faa,
  title        = {Air traffic control network dataset -- {KONECT}},
  howpublished = {KONECT: The Koblenz Network Collection},
  year         = {2017},
  url          = {http://konect.cc/networks/maayan-faa}
}

@misc{konect2017faa,
  title        = {Air traffic control network dataset -- {KONECT}},
  howpublished = {KONECT: The Koblenz Network Collection},
  year         = {2017},
  url          = {http://konect.cc/networks/maayan-faa}
}

@misc{faa2010,
  author       = {{Federal Aviation Administration}},
  title        = {Air Traffic Control System Command Center},
  year         = {2010},
  url          = {http://www.fly.faa.gov/}
}

@article{chen1985ranking,
  author  = {Chen, S.-H.},
  title   = {Ranking fuzzy numbers with maximizing set and minimizing set},
  journal = {Fuzzy Sets and Systems},
  volume  = {17},
  number  = {2},
  pages   = {113--129},
  year    = {1985}
}

@article{chanas1988single,
  author  = {Chanas, S. and Nowakowski, M.},
  title   = {Single value simulation of fuzzy variable},
  journal = {Fuzzy Sets and Systems},
  volume  = {25},
  number  = {1},
  pages   = {43--57},
  year    = {1988}
}

@article{dey2021imprecise,
  author  = {Dey, S. and Malakar, S. and Rajak, S.},
  title   = {Optimization of shortest path problem using Dijkstra's algorithm in imprecise environment},
  journal = {International Journal of Computer Applications Technology and Research},
  volume  = {10},
  number  = {10},
  pages   = {216--221},
  year    = {2021}
}

@article{ebrahimnejad2021mabc,
  author  = {Ebrahimnejad, A. and Enayattabar, M. and Motameni, H. and Garg, H.},
  title   = {Modified artificial bee colony algorithm for solving mixed interval-valued fuzzy shortest path problem},
  journal = {Complex \& Intelligent Systems},
  volume  = {7},
  pages   = {1527--1545},
  year    = {2021}
}

@article{elizabeth2012lambda,
  author  = {Elizabeth, J. and Sujatha, R.},
  title   = {Fuzzy shortest path problem based on level {$\lambda$}-triangular {LR} fuzzy numbers},
  journal = {Advances in Fuzzy Systems},
  volume  = {2012},
  number  = {1},
  pages   = {646248},
  year    = {2012}
}

@article{fan2025height,
  author  = {Fan, C.},
  title   = {A novel arithmetic operation function utilizing height of interval type-2 fuzzy sets to express risk preference},
  journal = {Ocean Engineering},
  year    = {2025},
  pages   = {122940}
}

@article{hassanzadeh2013genetic,
  author  = {Hassanzadeh, R. and Mahdavi, I. and Mahdavi-Amiri, N. and Tajdin, A.},
  title   = {A genetic algorithm for solving fuzzy shortest path problems with mixed fuzzy arc lengths},
  journal = {Mathematical and Computer Modelling},
  volume  = {57},
  number  = {1--2},
  pages   = {84--99},
  year    = {2013}
}

@article{hasuike2013robust,
  author  = {Hasuike, T.},
  title   = {Robust shortest path problem based on a confidence interval in fuzzy bicriteria decision making},
  journal = {Information Sciences},
  volume  = {221},
  pages   = {520--533},
  year    = {2013},
  doi     = {10.1016/j.ins.2012.09.025}
}

@article{kavian2012twostage,
  author  = {Kavian, Y. S. and Mahani, A. and Rashvand, H. F.},
  title   = {Two-stage uncertainty incorporating in optical core networks},
  journal = {IET Communications},
  volume  = {6},
  number  = {10},
  pages   = {1220--1228},
  year    = {2012}
}

@article{kumar2019shortest,
  author  = {Kumar, R. and Edalatpanah, S. A. and Jha, S. and Gayen, S. and Singh, R.},
  title   = {Shortest path problems using fuzzy weighted arc length},
  journal = {International Journal of Innovative Technology and Exploring Engineering},
  volume  = {8},
  number  = {6},
  pages   = {724--731},
  year    = {2019}
}

@article{pratibha2024leastcost,
  author  = {Pratibha and Dangwal, R.},
  title   = {Solving generalized fuzzy least cost path problem of supply chain network},
  journal = {Reliability Theory \& Applications},
  volume  = {19},
  number  = {4},
  pages   = {267--286},
  year    = {2024}
}

@article{sen2021similarity,
  author  = {Sen, S. and Patra, K. and Mondal, S. K.},
  title   = {Similarity measure of Gaussian fuzzy numbers and its application},
  journal = {International Journal of Applied and Computational Mathematics},
  volume  = {7},
  number  = {3},
  pages   = {96},
  year    = {2021}
}

@article{stoklasa2022relationship,
  author  = {Stoklasa, J. and Luukka, P. and Collan, M.},
  title   = {On the relationship between possibilistic and standard moments of fuzzy numbers},
  journal = {Journal of Computational and Applied Mathematics},
  volume  = {411},
  pages   = {114276},
  year    = {2022}
}

@article{tahayori2013shadowed,
  author  = {Tahayori, H. and Sadeghian, A. and Pedrycz, W.},
  title   = {Induction of shadowed sets based on the gradual grade of fuzziness},
  journal = {IEEE Transactions on Fuzzy Systems},
  volume  = {21},
  number  = {5},
  pages   = {937--949},
  year    = {2013}
}

@article{valdes2021fuzzy,
  author  = {Valdes, L. and Ariza, A. and Allende, S. M. and Trivi{\~n}o, A. and Joya, G.},
  title   = {Search of the shortest path in a communication network with fuzzy cost functions},
  journal = {Symmetry},
  volume  = {13},
  number  = {8},
  pages   = {1534},
  year    = {2021},
  doi     = {10.3390/sym13081534}
}

@article{verma2013greedy,
  author  = {Verma, M. and Shukla, K. K.},
  title   = {A Greedy Algorithm for Fuzzy Shortest Path Problem using Quasi-Gaussian Fuzzy Weights},
  journal = {International Journal of Fuzzy Systems and Applications},
  volume  = {3},
  number  = {2},
  pages   = {55--70},
  year    = {2013}
}

@article{yao2003fuzzy,
  author  = {Yao, J.-S. and Lin, F.-T.},
  title   = {Fuzzy shortest-path network problems with uncertain edge weights},
  journal = {Journal of Information Science and Engineering},
  volume  = {19},
  number  = {2},
  pages   = {329--351},
  year    = {2003}
}

\end{document}